\documentclass{article}
\usepackage{arxiv}
\usepackage{cite}
\usepackage{amssymb}
\usepackage{amsfonts}
\usepackage{amsmath}
\usepackage{tikz}
\usepackage{pgfplots}
\usepackage[T1]{fontenc}
\usepackage{subfigure}
\usepackage{mathtools}
\usepackage{array}
\usepackage{balance}
\usepackage{xspace}
\usepackage{tikz}
\usepackage{url}
\usepackage{tikzpeople}
\usepackage{varwidth}
\usepackage{algorithm}
\usepackage{algorithmicx}
\usepackage{algpseudocode}
\usepackage{lmodern}
\usepackage[british]{babel}
\usepackage[hidelinks]{hyperref}  
\usepackage{alltt}
\usepackage{footmisc}

\usetikzlibrary{positioning}
\pgfplotsset{compat=1.15}
\algrenewcommand\alglinenumber[1]{\scriptsize #1:}

\newcommand{\system}{{SEPAR}\xspace}

\newcommand{\ZERO}{\textsf{PREPARE}\xspace}
\newcommand{\zero}{{\sf \small prepare}\xspace}
\newcommand{\ONE}{\textsf{PROPOSE}\xspace}
\newcommand{\one}{{\sf \small propose}\xspace}
\newcommand{\TWO}{\textsf{ACCEPT}\xspace}
\newcommand{\two}{{\sf \small accept}\xspace}

\newcommand{\THREE}{\textsf{COMMIT}\xspace}
\newcommand{\three}{{\sf \small commit}\xspace}

\newcommand{\vchange}{{\sf \small failure-query}\xspace}

\newcommand{\newv}{{\sf \small new-primary}\xspace}

\newcommand{\TWOQ}{\textsf{ACCEPT-QUERY}\xspace}
\newcommand{\twoq}{{\sf \small accept-query}\xspace}

\newcommand{\twof}{{\sf \small super-accept}\xspace}

\newcommand{\sub}{{\sf \small submission}\xspace}
\newcommand{\Sub}{{\sf \small Submission}\xspace}

\newcommand{\clm}{{\sf \small claim}\xspace}
\newcommand{\clms}{{\sf \small claims}\xspace}

\newcommand{\ver}{{\sf \small verification}\xspace}

\newtheorem{theorem}{Theorem}
\newtheorem{lemma}{Lemma}
\newtheorem{metalemma}{Lemma}[section]
\newtheorem{Definition}[metalemma]{Definition}

\newenvironment{proof}{\noindent{\bf Proof:}\rm}

\newif\ifextend
\extendtrue

\newif\ifnextend
\nextendfalse

\title{SEPAR: Towards Regulating Future of Work Multi-Platform Crowdworking Environments with Privacy Guarantees}

 \author{
 Mohammad Javad Amiri$^1$ \quad Joris Dugueperoux$^2$ \quad Tristan Allard$^2$ \\
 {\bf Divyakant Agrawal$^3$ \quad Amr El Abbadi$^3$}\\
$^1$Department of Computer and Information Science, University of Pennsylvania\\
$^2$Univ Rennes, CNRS, IRISA\\
$^3$Department of Computer Science, University of California Santa Barbara\\
$^1$mjamiri@seas.upenn.edu, $^2$\{joris.dugueperoux, tristan.allard\}@irisa.fr, $^3$\{agrawal, amr\}@cs.ucsb.edu
\vspace{2em}
}

\begin{document}

\maketitle

\begin{abstract}
Crowdworking platforms provide the opportunity for diverse workers
to execute tasks for different requesters.
The popularity of the "gig" economy has given rise to independent platforms 
that provide competing and complementary services. 
Workers as well as requesters with specific tasks
may need to work for or avail from the services of multiple platforms 
resulting in the rise of {\em multi-platform crowdworking systems}.
Recently, there has been increasing interest by governmental, legal and social institutions
to enforce regulations,
such as minimal and maximal work hours,
on crowdworking platforms. 
Platforms within multi-platform crowdworking systems, therefore,
need to collaborate to enforce cross-platform regulations.
While collaborating to enforce global regulations
requires the {\em transparent} sharing of information about tasks and their participants,
the {\em privacy} of all participants needs to be preserved.
In this paper, we propose an overall vision exploring the
regulation, privacy, and architecture dimensions for the future of work multi-platform crowdworking environments.
We then present {\em \system}, a multi-platform crowdworking system that
enforces a large sub-space of practical global regulations on
a set of distributed independent platforms in a privacy-preserving manner.
\system, enforces {\em privacy} using {\em lightweight and anonymous tokens},
while {\em transparency} is achieved using fault-tolerant {\em blockchains} shared across multiple platforms.
The privacy guarantees of \system against covert adversaries are formalized and thoroughly demonstrated, 
while the experiments reveal the efficiency of \system in terms of performance and scalability.
\end{abstract}
\section{Introduction}\label{sec:intro}

The rise of the ''gig'' or \emph{platform economy}~\cite{cohen2017law,kenney2016rise}
is reshaping work all around the world. 
Crowdsourcing platforms dedicated to work
(also called \emph{crowdworking platforms}~\cite{ilocrowdwork18})
are online intermediaries between \emph{requesters} and \emph{workers},
where requesters propose \emph{tasks} while \emph{workers} propose skills and time.
By providing requesters (resp. workers) 24/7 access to a
worldwide workforce (resp. worldwide task market), crowdworking
platforms have grown in numbers, diversity,
and adoption\footnote{See for example:
Amazon Mechanical Turk (\url{https://www.mturk.com/}),
Wirk (\url{https://www.wirk.io/}), or
Appen (\url{https://appen.com/}) for micro-tasks,
Uber (\url{https://www.uber.com/}) or
Lyft (\url{https://www.lyft.com/}) for rides,
TaskRabbit (\url{https://www.taskrabbit.com/}) for home maintenance,
Kicklox (\url{https://www.kicklox.com/}) for collaborative engineering.}.
Today, crowdworkers come from countries
spread all over the world, and work on several, possibly competing,
platforms~\cite{ilocrowdwork18}. The use of crowdworking platforms is
expected to continue growing~\cite{ilofow19}, and in fact, they are
envisioned as key technological components of the future of work~\cite{fow, DBLP:journals/debu/BourhisDEHR19, DBLP:journals/debu/Gross-AmblardMT19}.

Crowdworking platforms, however, challenge national boundaries and weaken the
formal relationships between the platforms, workers and task requesters. 
Guaranteeing the compliance of crowdworking platforms with
national or regional labor laws is hard\footnote{See,
e.g., the \texttt{Otey V Crowdflower} class action against a famous microtask platform
for \emph{"substandard wages and oppressive working hours"}
(\url{casetext.com/case/otey-v-crowdflower-1}).}~\cite{ilofow19}
despite the stringent need for regulating work.
For example,
the preamble of the 1919 constitution of the International Labor
Organization~\cite{ilocons1919}, written in the ruins of World War I,
states that: \emph{``Whereas universal and lasting peace can be established
only if it is based upon social justice; (\ldots) an improvement of
those conditions is urgently required; as, for example, by the
regulation of the hours of work, including the establishment of a
maximum working day and week, the regulation of the labor supply
(\ldots)''.}
The total work hours of a worker per week may not exceed $40$ hours to
follow \emph{Fair Labor Standards Act}\footnote{\url{https://www.dol.gov/agencies/whd/flsa}} (FLSA).
In California, Assembly Bill 5 ({\em AB5})\footnote{
\url{https://leginfo.legislature.ca.gov/faces/billTextClient.xhtml?bill_id=201920200AB5}}
entitles workers to greater labor protections, such as minimum wage laws,
sick leave, and unemployment and workers' compensation benefits.
{\em AB5} is currently being challenged by {\em California Proposition 22}\footnote{ \url{https://ballotpedia.org/California_Proposition_22,_App-Based_Drivers_as_Contractors_and_Labor_Policies_Initiative_(2020)}},
which also imposes its own set of regulations on minimal hours worked for health benefits. 
The global regulation of the work hours represents the {\em minimal} and {\em maximal}
number of hours that participants, i.e., worker, requester, and platform,
can spend on crowdworking platforms.
While legal tools are currently being investigated~\cite{ilofow19,cwcnnum20},
there is a stringent need for technical tools allowing official institutions to enforce regulations.

Some platforms have already started implementing self-defined \emph{local regulations}.
For example, Uber\footnote{ \url{https://www.uber.com/en-ZA/blog/driving-hours-limit/}} and
Lyft\footnote{\url{https://help.lyft.com/hc/articles/115012926787-Taking-breaks-and-time-limits}}
force drivers to rest at least 6 hours for every 12 hours in driver mode, or
Wirk.io\footnote{\url{https://www.wirk.io/50k-freelances-en-france/}}
prevent micro-workers from earning more than $3000$ euros per year.
However, since workers and requesters can simply switch platforms when a local limit is reached,
no global cross-platform regulation can be enforced.
Moreover, participants in a crowdworking task
may also behave maliciously or act as adversaries,
e.g., violate the privacy of participants or the regulations for their benefits. %
The privacy of participants and the global consistency of regulations are considered as critical needs for future of work crowdworking environments~\cite{DBLP:journals/debu/BourhisDEHR19}.

Most current crowdworking platforms are independent of each other.
However, the emergence of complex tasks that may need multiple contributions
from possibly different platforms,
on one hand, and more importantly, the enforcement of
legal regulations, on the other hand, highlights the need for
collaboration between crowdworking platforms,
thus resulting in {\em multi-platform crowdworking systems}.
For example, many drivers work for both Uber and Lyft concurrently\footnote{ For example,
\texttt{ridesharapps.com}
provides tutorials to help drivers manage apps to optimize their
earnings~\url{https://rideshareapps.com/drive-for-uber-and-lyft-at-the-same-time/}.},
while a requester may also request multiple rides from both Uber and
Lyft concurrently.
The observation holds also for
on-demand services\footnote{See, e.g., \url{https://tinyurl.com/nytgigmult}.},
as well as microtask platforms~\cite{ilocrowdwork18}
where a requester who has registered on \emph{Amazon Mechanical Turk} and \emph{Prolific}
might need hundreds of contributions for a single microtask at the same time and
accept these contributions from workers
regardless of the platforms the microtasks are performed on. 
Since workers from different platforms might want to perform these contributions,
the system needs to establish consensus among the various microtask platforms
to assign workers and provide the specified number of contributions while
ensures minimal and maximum hourly regulations on workers
without revealing any private information about the workers to competing platforms.

Our overall vision for multi-platform crowdworking environments needs to address three main dimensions:
{\em regulations}, {\em privacy}, and {\em architecture}.
First, a multi-platform crowdworking system must clearly define the \emph{types of regulations} supported
in terms of the \emph{complexity} of the regulation (e.g., a simple or a chain of interactions) as well as
the \emph{aggregate} requirements of the regulation (e.g., no aggregation or \texttt{SUM} aggregations). 
For example, California Proposition 22,
which states "if a driver works at least $25$ hours per week, companies (i.e., platforms)
require to provide healthcare subsidies \ldots",
is a simple, \texttt{SUM}-aggregate regulation.
Second, the threat model (e.g., \emph{honest-but-curious}, \emph{covert}, \emph{fully malicious})
as well as the 
privacy guarantees provided to each entity must be clearly specified.
Finally, from an architecture design point of view,
any multi-platform crowdworking environment consists of two critical components:
{\em regulation management} that models and enforces regulations and
{\em global state management} that manages the global states of all participants as well as tasks.
Each of these components can be implemented in
a centralized or decentralized approach.

In this paper, we present our overall vision for future of work multi-platform crowdworking environments together with {\em \system},
a possible instance of a precise point in the design space of multi-platform crowdworking systems.
\system results from choices, guided by cutting-edge real-life regulation proposals,
on all three regulations, privacy and architecture dimensions of the design space.
First, \system focuses on managing lower and upper bounds on aggregate-based regulations.
Second, \system considers that any participant in a crowdworking environment
may act as a \emph{covert adversary}~\cite{aumann2007security} and
ensures that no participant obtains or infers any information
about a crowdworking task beyond what is strictly needed for accomplishing
its local crowdworking task and for the distributed enforcement of regulations.
Finally, in \system, we opted for simplicity and rapid prototyping to use
a centralized (but fault-tolerant) component to manage regulations, however,
a decentralized component to manage the global state of the system.
In particular, \system uses a {\em permissioned blockchain}
as an underlying infrastructure that is shared among all involved platforms.

The complexity of the conjunction of the required properties
makes the problem non-trivial.
First, the regulations need to be expressed in a {\em simple} and {\em non-ambiguous} manner.
Second, while enforcing regulations over multiple crowdworking platforms
requires the global state of the system to be {\em transparent},
the {\em privacy} of participants needs to be preserved,
hence \system needs to reconcile {\em transparency} with {\em privacy}.
Finally, the decentralized management of the global state among a distributed set of crowdworking platforms
requires distributed consensus protocols.

\system is a two-level solution consisting of a \emph{privacy-preserving token-based system}
(i.e., the application level)
on top of a \emph{blockchain ledger}
shared across platforms (i.e., the infrastructure level).
First, at the application level, global regulations are modeled
using {\em lightweight and anonymous tokens} distributed to workers,
platforms, and requesters. 
The information shared among platforms and participants is limited to the minimum
necessary for performing the tasks against
\emph{adversarial participants acting as covert adversaries}.
Second, at the infrastructure level, the blockchain ledger
allows \system to provide transparency across platforms.
Nonetheless, for the sake of privacy and to improve performance,
the ledger is {\em not maintained} by any platform and
each platform maintains only a view of the ledger.
We then design a suite of consensus protocols for coping with the concurrency issues
inherent to a multi-platform context.
Salient features of \system include the simplicity of its building
blocks (e.g., usual signature schemes) and its compatibility
with today's platforms (e.g., it does not jeopardize their privacy
requirements of requesters and workers for enforcing the regulation).

In a nutshell, the contributions of this paper are as follows:
\begin{enumerate}
\item A vision for the design space of regulation systems for future of work multi-platform crowdworking environments. In particular, we (1) express regulations as \texttt{SQL} constraints and categorize them according to their \texttt{SQL} expression, (2) propose a formal privacy model for multi-platform regulated crowdworking systems 
based on the well-known simulatability paradigm, and (3) discuss critical components of their architecture. 
\item \system, a two-level {\em privacy-preserving transparent} multi-platform proof-of-concept crowdworking system that enforces a precise point in the design space guided by cutting edge practical regulations currently discussed by societal organizations, legal entities and enterprises. \system uses a simple language for expressing {\em global regulations} and mapping them to \texttt{SQL} constraints to ensure semantic clarity. It ensures privacy using \emph{lightweight and anonymous tokens}, while transparency is achieved using a \emph{blockchain} shared across platforms for both crash-only and Byzantine nodes.
\item A formal security analysis of \system and thorough experimental evaluations.
\end{enumerate}

The paper is organized as follows.
The overall vision of the design space for regulation systems is presented in Section~\ref{sec:designspace}.
Section~\ref{sec:model} presents the \system model within the design space.
The implementation of regulations in \system is discussed in Section~\ref{sec:enf}.
The blockchain ledger and consensus protocols are presented in Section~\ref{sec:dist}.
Section~\ref{sec:eval} details an experimental evaluation,
Section~\ref{sec:related} discusses the related work, and
Section~\ref{sec:conc} concludes the paper.
\section{Design Space of Regulation Systems}
\label{sec:designspace}

Designing a system for regulating crowdworking platforms mandates three main choices. First, given the large variety of regulations that apply to working environments, any given regulation system must clearly define the \emph{types of regulations} it supports. Second, a crowdworking environment involves a set of distributed entities (platforms, workers, requesters) that cannot be fully trusted. The \emph{privacy guarantees} provided to each entity must be rigorously stated (e.g., computational guarantees). Third, the distributed nature of a crowdworking environment requires various \emph{architectural choices} that range from fully decentralized approaches (e.g., peer-to-peer architectures) to the traditional centralized architecture.
In this section, we first define the crowdworking environment that we consider and then
characterize 
the design space for crowdworking regulation systems. 

\subsection{Crowdworking Environment}

A \emph{crowdworking environment} consists of
a set of workers $\mathcal{W}$ interacting with a set of requesters $\mathcal{R}$
through a set of competing platforms $\mathcal{P}$.
We refer to the workers, platforms, and requesters of a crowdworking environment as \emph{participants}.
Each worker $w \in \mathcal{W}$
(1) registers to one or more platforms $\mathcal{P}_{w} \subset \mathcal{P}$
according to her preferences and, through the latter,
(2) accesses the set of tasks available on $\mathcal{P}_{w}$,
(3) submits each {\em contribution} to the platform $p \in \mathcal{P}_{w}$ she elects, and
(4) obtains a {\em reward} for her work.
On the other hand, each requester $r \in \mathcal{R}$ similarly
(1) registers to one or more platforms $\mathcal{P}_{r} \subset \mathcal{P}$,
(2) issues a {\em submission} which contains her tasks $\mathcal{T}_{r}$
to one or more platforms $p \in \mathcal{P}_{r}$,
(3) receives the contributions of each worker $w$ registered to
$\mathcal{P}_{r} \cap \mathcal{P}_{w}$ having elected a task $t \in \mathcal{T}_{r}$, and
(4) launches the distribution of rewards.
Platforms are thus in charge of facilitating the interactions between workers and requesters.
A \emph{crowdworking process} $\pi$ connects three participants
-- a workers $w$, a platform $p$, and a requester $r$ --
and aims to facilitate the execution of a task $t \in \mathcal{T}_{r}$ through platform $p$ via a worker $w$.
For simplicity and without loss of generality, we assume that each process corresponds
to a time unit of work (e.g., $1$ hour). 
Figure~\ref{fig:arch} shows a crowdworking infrastructure with
four platforms, four workers and four requesters.

\begin{figure}[t]
\begin{center}
\includegraphics[width= 0.4\linewidth]{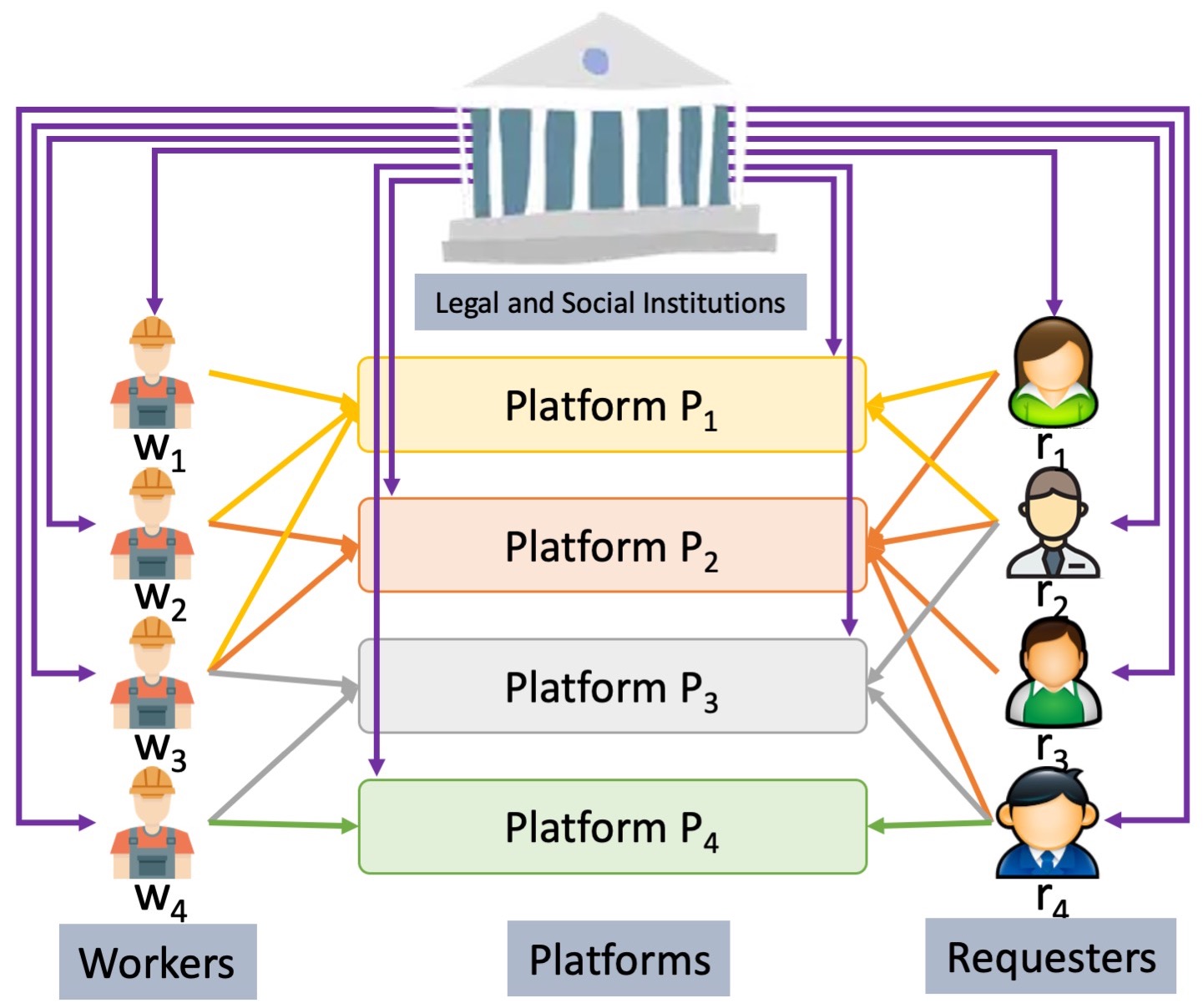}
\caption{A crowdworking infrastructure}
\label{fig:arch}
\end{center}
\end{figure}

\subsection{Types of Regulation}
\label{sec:expcert}

The space of possible regulations can be structured based on two orthogonal dimensions:
the \emph{complexity} of the regulation, e.g.,
constraints on a single process versus constraints that apply to multiple processes that are transitively related,
and the \emph{aggregate} nature of the regulation, e.g., no aggregation versus restrictions involving \texttt{SUM} aggregate.

In order to give a clear semantics to each dimension and to rigorously map regulations to points in this space, we express regulations by \texttt{SQL} constraints over a \emph{universal virtual table} storing information about all crowdworking processes having been performed. We denote this table by \texttt{U-TABLE} and emphasize that it is \emph{virtual} (we use it only for clarifying the various types of regulations, it is never instantiated). The attributes of \texttt{U-TABLE} refer to meta-data about the interactions (i.e., at least the worker, requester, and platform involved in a process, and also possibly additional metadata such as begin and end timestamps) and information about the contents (e.g., the time estimate for the task\footnote{Future regulation systems will need to design technical means to guarantee the reliability of the time estimates for tasks (e.g., privacy-preserving feedback systems from workers, automatic time estimation by analyzing task descriptions).} - the \texttt{TIMECOST} attribute below -, the proposed wage, the worker's contribution). For simplicity, we focus below on the attributes of the \texttt{U-TABLE} relevant to our illustrative examples: (1) \texttt{TS\_BEGIN}, \texttt{TS\_END}, \texttt{WORKER}, \texttt{PLATFORM}, and \texttt{REQUESTER}, and (2) \texttt{TIMECOST}, \texttt{WAGE}, 
and \texttt{CONTRIBUTION}. 
Finally, a regulation is simply a Boolean \texttt{SQL} expression nested within the usual \texttt{CHECK} clause: \verb|ALTER TABLE U ADD CONSTRAINT r CHECK ( ... )|. 

The complexity and aggregate dimensions of a regulation can both be deduced from its \texttt{SQL} expression. The complexity is given by the presence of {\em join} operations while the aggregate dimension is given by the presence of {\em aggregate} function(s), possibly with \texttt{GROUP BY} and \texttt{HAVING} clause(s). For simplicity, we consider a coarse grain characterization of these two dimensions. A regulation is \texttt{simple} if there is no join and \texttt{complex} otherwise. A regulation is \texttt{row-only} if it does not involve any aggregate function, \texttt{aggregate-only} if it involves only comparison(s) over aggregate(s), and \texttt{mixed} if it involves comparisons over rows and aggregates.   

We illustrate the possible types of regulations based on simple examples extracted from real-life crowdworking regulations or from real-life proposals of regulation. First, we consider a regulation $r_1$ requiring the wage proposed by each task to be at least a given amount $\theta$. This regulation is similar to CA Proposition 22. It illustrates the (\texttt{simple}, \texttt{row-only}) type of regulation. 
\begin{alltt}
{\small
ALTER TABLE U-TABLE ADD CONSTRAINT \(r1\) CHECK (
  NOT EXISTS ( 
    \textbf{SELECT} * \textbf{FROM} U
    \textbf{WHERE} TIMECOST \(\leq \theta\)
  ) );
}
\end{alltt}

Second, we consider a regulation $r_2$ requiring each worker to work at most a given amount of time units $\theta$ per time period $\rho$. It illustrates the (\texttt{simple}, \texttt{mixed} with \texttt{SUM}-aggregate) type of regulation. It is similar to the regulation followed by the \texttt{wirk.io} platform that limits the crowdworking gains of any worker on the platform to 3000 euros per year. The following \texttt{SQL} constraint expresses $r_2$, assuming that \texttt{current\_time()} gives the current time in the same unit as the period $\rho$.
\begin{alltt}
{\small
ALTER TABLE U-TABLE ADD CONSTRAINT \(r2\) CHECK (
  NOT EXISTS ( 
    \textbf{SELECT} * \textbf{FROM} U
    \textbf{WHERE} WORKER=\(w\) AND \texttt{current_time()}-\texttt{TS\_BEGIN} \(\leq \rho\)
    \textbf{GROUP BY} WORKER
    \textbf{HAVING} \texttt{SUM}(TIMECOST) \(\geq \theta\)
  ) ); 
}
\end{alltt}

Finally, we complete our illustrations by considering a regulation $r_3$ that prevents any worker to submit two similar contributions to the same requester (even through two distinct platforms). We assume that the \texttt{sim} function computes the similarity between two contributions and that $\theta$ is the threshold above which we consider that two contributions are too similar. This illustrates the (\texttt{complex}, \texttt{row-only}) type of regulation. 
\begin{alltt}
{\small
ALTER TABLE U-TABLE ADD CONSTRAINT \(r3\) CHECK (
  NOT EXISTS ( 
    \textbf{SELECT} * 
    \textbf{FROM} U U1 JOIN U U2 ON 
        U1.WORKER=U2.WORKER 
        AND U1.REQUESTER=U2.REQUESTER
        AND \texttt{sim}(U1.CONTRIBUTION, U2.CONTRIBUTION) \(\geq \theta\)
  ) );
}
\end{alltt}

Although most regulations must always hold (e.g., a lower than constraint on a (\texttt{simple}, \texttt{SUM}-aggregate) regulation, a (\texttt{complex}, \texttt{row-only}) regulation), some regulations, inherently, \emph{cannot} always hold. Similar to \texttt{deferred} \texttt{SQL} constraints, they must only hold after a given time period. For example, a periodic greater than constraint on a (\texttt{simple}, \texttt{SUM}-aggregate) regulation cannot hold initially, but must hold at the end of a given period (e.g., CA Proposition 22 that requires a worker to work at least 25 hours per week to qualify for healthcare subsidies). We call \emph{enforceable} the regulations that must always hold and \emph{verifiable} the regulations that eventually hold. The verifiable/enforceable property of a regulation is only due to its nature not to its implementation (but it impacts it directly). Future crowdworking regulation systems need to determine the enforceable/verifiable properties of the regulations
they support.

\subsection{Threat Model and Privacy Model}
Crowdworking environments are open environments that connect possibly adversarial participants. Any regulation system must clearly specify both the threat model (e.g., \emph{honest-but-curious}, \emph{covert}, \emph{fully malicious}) and the privacy model. %
While the threat model only depends on the underlying system, the privacy model can be based on a common formal requirement parameterized for each system by the leaks it tolerates. The possible leakage ranges from no information leaked to any participant (similar to usual \emph{secure multi-party computation} algorithms) to full disclosure (of all the information discussed above) to all participants (e.g., in current crowdworking platforms, the underlying system often requires full disclosure to the platform). The disclosures tolerated by future regulation systems will fall within this range. 

We formalize below a common privacy model based on the well-known simulatability paradigm often used by secure multi-party computation algorithms. The proposed model guarantees that nothing leaks, with computational guarantees\footnote{The majority of data protection techniques used in real-life are based on encryption schemes that provide guarantees against computationally-bounded adversaries but the model can be easily adapted to information theoretic attackers.}, except the \emph{pluggable} system-dependent tolerated \emph{disclosures}. 

Consider a crowdworking process $\pi$ between worker $w$, platform $p$, and requester $r$ for solving a task $t$. The task may have been sent to several platforms and, depending on the underlying crowdworking platforms, might have been accessed by several workers before being picked and solved. The information generated by the execution of $\pi$ consists, similarly to the information captured by the \texttt{U-TABLE}, of information about interactions (e.g., at least $w$, $r$, and $p$, the participants directly involved in $\pi$) and information about contents (e.g., description of $t$, contribution, proposed wage). We propose to define the sets of disclosure according to the involvement of a participant in $\pi$. Indeed, the platforms not involved in $\pi$ (i.e., different from $p$) but that received $t$ may need to learn that $t$ has been completed (e.g., to manage their local copy of the task).  %
Similarly, workers who are not involved in a task $t$ might still need to know that it has been executed, while potentially preserving the privacy of worker $w$ who executed the task.
The three sets of disclosures defined below cover all these cases\footnote{This can be extended by tuning the disclosure sets (e.g., by distinguishing the requesters from the platforms in the set of involved participants).}. Future regulation systems have to specify clearly the content of each disclosure set.
\begin{itemize}
  
\item Disclosures to the participants that are \textbf{not involved} in $\pi$ and that have \textbf{not received} task $t$ from requester $r$: $\delta_{\neg R \neg I}^\pi$ 
    
\item Disclosures to the platforms and workers that have \textbf{received} the task $t$ from $r$ but that are \textbf{not involved} in $\pi$: $\delta_{R \neg I}^\pi$ 
    
\item Disclosures to the participants that are directly \textbf{involved} in $\pi$ (and have thus \textbf{received} task $t$): $\delta_{RI}^\pi$ 
\end{itemize}

\begin{Definition}
  Let $\Pi$ be a set of crowdworking processes executed by $\varsigma$
  an instance of \system over a set of participants. We say that
  $\varsigma$ is \emph{$\delta^\Pi$-Private} if, for all
  $\pi \in \Pi$, for all computationally-bounded adversaries
  $\mathtt{A}$, the sets of disclosures
  $(\delta_{\neg R \neg I}^\pi, \delta_{R \neg I}^\pi, \delta_{RI}^\pi)$, 
  assuming arbitrary background knowledge $\chi \in \{0, 1\}^*$,
  the distribution representing the adversarial
  knowledge over the input dataset in the real setting is
  \emph{computationally indistinguishable} from the distribution
  representing the adversarial knowledge in an ideal setting in which
  a trusted third party \texttt{cp} executes the crowdworking
  process $\pi$ of $\varsigma$:
  \ifextend
\begin{displaymath}
  \mathtt{REAL}_{\varsigma, \mathtt{A} (\chi, \delta_i^{\pi})}
  (\mathcal{W}, \mathcal{P}, \mathcal{R}, \mathcal{T}) 
  \stackrel{c}{\equiv} 
  \mathtt{IDEAL}_{\mathtt{cp}, \mathtt{A}
    (\chi, \delta_i^{\pi})} (\mathcal{W}, \mathcal{P}, \mathcal{R},
  \mathcal{T})
\end{displaymath}
\fi
\ifnextend
$\mathtt{REAL}_{\varsigma, \mathtt{A} (\chi, \delta_i^{\pi})}
  (\mathcal{W}, \mathcal{P}, \mathcal{R}, \mathcal{T}) 
  \stackrel{c}{\equiv} 
  \mathtt{IDEAL}_{\mathtt{cp}, \mathtt{A}
    (\chi, \delta_i^{\pi})} (\mathcal{W}, \mathcal{P}, \mathcal{R},
  \mathcal{T})$
\fi
where $i \in \{\neg R \neg I, R\neg I, RI\}$, and $\mathtt{REAL}$ denotes the adversarial
knowledge in the real setting and $\mathtt{IDEAL}$ its counterpart in
the ideal setting.
\end{Definition}

\subsection{Architecture}

The architecture of a multi-platform crowdworking system consists of two main building blocks:
{\em Regulation Management} and {\em Global State Management}.
Regulation management {\em models} the regulations among the participants and
{\em ensures} that the modeled regulations are adhered to by all participants.
Global state management, on the other hand, stores the global states of the system including
all information related to the participants and tasks.
To implement these two components, similar to all distributed systems,
either a centralized or a decentralized approach can be employed.
Centralization is typically easier to rapid prototype,
while requiring additional technologies to ensure fault-tolerance, privacy, and trustworthiness.
A decentralized approach, on the other hand,
is more compatible with the setting consisting of diverse organizational entities,
while resulting in more overhead and complex communication protocols among entities.

The verifiable/enforceable property of a regulation is another architectural challenge.
While enforceable regulations 
could be enforced within the multi-platform system,
verifiable regulations might be of interest to an outside entity,
e.g., legal courts or insurance companies. 
In the latter case, the system needs to
provide evidence to an outside entity demonstrating that
the regulation was adhered to and hence resolve any disputes that may arise.
\section{\system: Design Choices}
\label{sec:model}
We propose \system as a possible instance of a precise point in the design space of regulation systems. \system results from choices, guided by cutting-edge real-life regulation proposals, on the three dimensions of the design space: supported regulations, disclosures tolerated, and  architecture. 

\subsection{Supported Regulations}
\system focuses on enforcing lower and upper bounds on the aggregated working time spent on crowdworking platforms\footnote{ Regulating the wages earned through crowdworking platforms can be dealt with similarly.}, a consensual societal need. The necessary information are the participants to crowdworking processes and the (discrete) time estimation of tasks\footnote{ Extending regulations with validity periods (e.g., "one week"), is straightforward.}: for \system, the \texttt{U-TABLE} thus consists of the \texttt{WORKER}, \texttt{PLATFORM}, \texttt{REQUESTER}, and \texttt{TIMECOST} attributes. \system does not consider joins. As a result the kind of regulations supported by \system is (\texttt{simple}, \texttt{mixed} with \texttt{SUM}-aggregate). 
The enforceable/verifiable nature of the regulations supported by \system are easy to determine:
lower-than regulations are enforceable because the upper-bound guarantee must always hold, and
greater-than regulations are verifiable because the lower-bound guarantee cannot always hold (but will have to hold finally, e.g., at the end of each time period).

For simplicity, we propose to express such regulations by (1) a triple $(w, p, r)$ that associates a worker $w$, a platform $p$, and a requester $r$, (2) a comparison operator $<$ or $>$, and (3) a threshold value $\theta$ (an integer) that defines the lower/upper bound that needs to hold.  Intuitively, a regulation $((w, p, r), <, \theta)$ states %
that there must not be more than $\theta$ time spent by worker $w$ on platform $p$ for requester $r$.
We also allow two wildcards to be written in any position of a triple: $*$ and $\forall$. First, the $*$ wildcard allows to ignore one or more elements of a triple\footnote{ Intuitively, the $*$ wildcard means "whatever".}. For example $(*, p, r)$ means that the regulation applies to the pair $(p, r)$. A triple may contain up to three $*$ wildcards. 
An element of a triple that is not a $*$ wildcard is called a \emph{target} of the regulation. 
Second, the $\forall$ wildcard
allows a regulation to express a constraint that must hold for all participants in the same group of participants\footnote{ Intuitively, the $\forall$ wildcard means "for each".}. For example, $(\forall, p, r)$ represents the following set of triples: $\{(w, p, r)\}$, $\forall w \in \mathcal{W}$. %
As a result, enforceable regulations are expressed by $((w, p, r), <, \theta)$ tuples (possibly with wildcards), and verifiable regulations by $((w, p, r), >, \theta)$ tuples (possibly with wildcards). 

\textbf{Examples.} The semantics of an enforceable regulation without any wildcard, e.g., $e \leftarrow ((w, p, r), <, 26)$ expressing a higher bound on the number of time units spent (i.e., 25 time units) 
by worker $w$ for requester $r$ on platform $p$, is the same as the following \texttt{SQL} query: %
\begin{alltt}
{\small
ALTER TABLE U-TABLE ADD CONSTRAINT \(e\) CHECK (
  NOT EXISTS ( 
    \textbf{SELECT} * \textbf{FROM} U-TABLE
    \textbf{WHERE} WORKER=\(w\) AND PLATFORM=\(p\) AND REQUESTER=\(r\)
    \textbf{GROUP BY} WORKER, PLATFORM, REQUESTER
    \textbf{HAVING} SUM(TIMECOST) \(\geq 26\)
  ) );
}
\end{alltt}

\noindent
The weekly FLSA limit on the total work hours per worker can easily be expressed as $((\forall, *, *), <, 40)$. A social security institution can request each worker $w$ applying for insurance coverage to prove that she worked more than $5$
hours: $((w, *, *), >, 5)$ is both necessary and sufficient. Similarly, the regulation $((*, p, *), >, 1000)$ allows a tax institution to require from each platform $p$ applying for a tax refund that the total work hours of all its workers is at least $1000$ hours.

\subsection{Threat Model and Disclosure Sets}
\system considers that any participant in a crowdworking environment (worker, requester, platform) may act as a \emph{covert adversary}~\cite{aumann2007security} that aims at inferring anything that can be inferred from the execution sequence and that is able to deviate from the protocol if no other participant detects it. For simplicity, we assume in \system that adversarial participants do not collude (although extending \system to cope with colluding covert adversaries is easy (see Section~\ref{sec:enf})).

Consider a crowdworking process $\pi$ between worker $w$, platform $p$, and requester $r$ for solving task $t$. The information generated by the execution of $\pi$ consists of the relationship between the three participants ($w$, $p$, and $r$) with the task $t$. Additionally, we also consider the information that $\pi$ is starting or ending through a starting event \texttt{BEGIN} and an ending event \texttt{END}. They may be determined, for example, by exchanging messages among the participants of $\pi$, and may include additional concrete information, e.g., timestamps, IP address.
$(\mathtt{BEGIN}, \mathtt{END}, w, p, r, t)$ denotes the information generated by $\pi$.
\system does not leak any information about the worker and the requester involved in $\pi$ when it is not needed by $\pi$.
It tolerates the disclosure of the \{\texttt{BEGIN}, \texttt{END}\} events and of the platform $p$ to all participants, whatever their involvement in $\pi$. This allows platforms to share information for enforcing regulations (e.g., check that all participants satisfy the regulations before executing $\pi$), and to collaborate for correctly managing cross-platforms tasks.
Note that for simplicity we use the same notation $\delta$ for disclosures concerning \emph{sets} of crowdworking processes as well.

The resulting disclosure sets of \system are instantiated as follows:
\begin{itemize}

\item The participants that are \textbf{not involved} in $\pi$ and that have \textbf{not received} task $t$ from requester $r$ must not learn anything about the worker, the task, and the requester involved in $\pi$: $\delta_{\neg R \neg I}^\pi=(\mathtt{BEGIN}, \mathtt{END}, p)$

\item The platforms and workers that have \textbf{received} task $t$ from $r$ but that are \textbf{not involved} in $\pi$ must be aware that $t$ has been performed (e.g., for not contributing to  $t$) but must not know that it has been performed by worker $w$: $\delta_{R \neg I}^\pi = (\mathtt{BEGIN}, \texttt{END}, p, r, t)$

\item The participants that are directly \textbf{involved} in $\pi$ (and have thus \textbf{received} task $t$) learn the complete 6-tuple: $\delta_{RI}^\pi = (\mathtt{BEGIN}, \mathtt{END}, w, p, r, t)
$
\end{itemize}

\subsection{Architecture}
\label{sec:archi}

\system has two main components.
The centralized {\em Registration Authority (RA)} and the decentralized {\em Multi-Platform Infrastructure (MPI)}.
Although centralized, the RA can be made fault-tolerant using standard replication \cite{lamport1978implementation} techniques.
RA registers the participants to the crowdworking environment, models the regulations,
and distributes to participants the cryptographic material necessary for enforcing or verifying regulations in a secure manner.

MPI, on the other hand, is a decentralized component that
maintains the global state of the system including all operations performed by the participants.
This state is maintained within a distributed persistent transparent {\em ledger}.
MPI consists of a set of collaborating crowdworking platforms connected by an asynchronous distributed network.
Due to the unique features of blockchains such as
transparency, provenance, and fault tolerance, the MPI is implemented as a {\em permissioned blockchain}.

MPI processes two types of tasks: {\em internal} and {\em cross-platform}.
Internal tasks are submitted to the ledger of a single platform, whereas
cross-platform tasks are submitted to the ledgers of multiple platforms (i.e, involved platforms).
\system does not make any assumptions on the implementation of crowdworking processes by platforms,
e.g., task assignment algorithm and workers contribution delivery.
However, processing a task (either internal or cross-platform) requires agreement from the involved platform(s).
To establish agreement among the nodes within or across platforms,
\system uses {\em local} and {\em cross-platform} consensus protocols. 
Furthermore, in both internal and cross-platform tasks,
\system enables all platforms to check the fulfillment of regulations
using a {\em global} consensus protocol  among {\em all} platforms
(irrespective of whether they are involved in the execution of task).
\section{\system: Implementation of the Regulations}
\label{sec:enf}

Inspired by e-cash systems, \system implements both enforceable and verifiable regulations by managing two \emph{budgets} per participant while guaranteeing both privacy and correctness. The overall system is conceptually simple and only relies on the correct use of indidivual and group signatures.
The registration authority (RA, see Section~\ref{sec:archi}) bootstraps \system, its participants, and refreshes their budgets periodically, but is not involved in the continuous execution of crowdworking processes and their regulations.

\subsection{Individual and Group Signatures}
Workers, requesters, and platforms are all equipped by the registration authority with the following cryptographic material: a pair of {\em individual} public/private asymmetric keys (e.g., RSA) and a pair of public/private asymmetric \emph{group} keys (e.g., \cite{camenisch2004group}) where the union of all workers forms a group (in the sense of group signatures), the union of all requesters forms another group, and the union of all platforms forms the last group. 
A group signature scheme respects three main properties~\cite{chaum1991group}:
(1) only members of the group can sign messages, %
(2) the receiver of the signature can verify  that it is a valid signature for the group but cannot discover which member of the group computed it, and %
(3) in case of dispute later on, the signature can be "opened" (with or without the help of the group members) to reveal the identity of the signer. %
A common way to enforce the third property is to rely on a \emph{group manager} that can add new members to the group or revoke the anonymity of a signature. Instances of such schemes are proposed in~\cite{chaum1991group}, but also in~\cite{ateniese2000practical,camenisch2004group}. In \system, we use the protocol proposed in~\cite{camenisch2004group} and denote
$\sigma^g_w(m)$ 
the group signature of the worker $w$ (with her group private key) for message $m$. We use equivalent notations for requester $r$ and platform $p$ (i.e., respectively $\sigma^g_r(m)$ and $\sigma^g_p(m)$). %
The notation %
$\sigma^i_w(m)$ %
is used to refer to a individual asymmetric signature (e.g., RSA) of worker $w$ (with her individual (non group) private key) for the message $m$. We use equivalent notations for requester $r$ and platform $p$ (i.e., respectively $\sigma^i_r(m)$ and $\sigma^i_p(m)$). 
The registration authority is the group manager for the three groups (workers $\mathcal{W}$, requesters $\mathcal{R}$, platforms $\mathcal{P}$) and is also equipped with her own individual cryptographic keys for signing her messages : $\sigma^i_{RA}(m)$. %

\subsection{A Simple Token-Based System}
\label{sec:const_implem}

To regulate crowdworking processes, our token-based system is defined by five functions:
\texttt{GENERATE} for initializing the budgets with the correct number of tokens and refilling them, %
\texttt{SPEND} for spending portions of the budgets, %
\texttt{PROVE} for providing proofs for verifying verifiable regulations (e.g., to a third party), %
\texttt{CHECK} for checking whether a given spending is allowed or not, and %
\texttt{ALERT} for reporting dubious spending. %
Since the execution of these functions
changes the global state of the system,
the data involved in the execution (e.g., tasks, tokens, signatures)
must be appended to the distributed ledger of the platforms.

\subsubsection{The GENERATE Function.}
The registration authority uses the \texttt{GENERATE} function to create tokens for all participants (i.e., workers, platforms, requesters) according to the set of regulations. We call {\em $e$-tokens} the tokens implementing enforceable regulations and
{\em $v$-tokens} the tokens implementing verifiable regulations\footnote{ Although it is technically possible to use a unified token structure for both enforceable and verifiable regulations, using two different tokens reduces computation and communication costs.}. %

\noindent
\textbf{$e$-tokens.}
For each enforceable regulation $((w, p, r), <, \theta+1)$, the registration authority generates $\theta$ $e$-tokens and sends a copy of each token to each target of the regulation. An $e$-token consists of a \emph{public component} - i.e., a pair made of a \emph{number used only once} (referred to as a {\em nonce} below) generated by the registration authority and a signature of the nonce by the registration authority\footnote{ Extending tokens with labels and timestamps to support validity periods is straightforward.} -
and a \emph{private component} - an index for selecting the correct set of tokens given the other targets, i.e., their public keys\footnote{ The use of a public key generated by the registration authority is important here because (1) it can be shared among participants without disclosing their identities, i.e., it is a pseudonym, (2) the corresponding private key can be used by participants for mutual authentication in order to guarantee the correctness of the index and consequently of the choice of tokens.}.
Let e-$\tau$ be an $e$-token, e-$\tau_{pub}$ be its public component, e-$\tau_{priv}$ be its private component, $N$ be a nonce
and $\lambda$ the list of public keys of the targets of the corresponding regulation. The $e$-token is thus the pair $($e-$\tau_{pub}, $e-$\tau_{priv})$ where e-$\tau_{pub} = (N, \sigma^i_{RA}(N))$ and e-$\tau_{priv} = \lambda$. %

\medskip\noindent
\textbf{$v$-Tokens.}
Similar to an $e$-token, a $v$-token consists of a public and a private component. The public component is a nonce together with its signature from the registration authority. The private component is simply a triplet of signatures, from the registration authority, binding the recipient of the $v$-token (called \emph{owner} below) to each of the other targets of the verifiable regulation\footnote{ Binding recipients to their $v$-tokens allows to trace possible malicious leakages of $v$-tokens (they strongly bind participant to a crowdworking process, contrary to $e$-tokens) and consequently to prevent them (our participants act as covert adversaries)}. More formally, let v-$\tau$ be a $v$-token, v-$\tau_{pub}$ be its public component, v-$\tau_{priv}$ be its private component, $N$ be a nonce, $o$ be the identity of the participant owner of the token, and $(w, p, r)$ be the related triplet. The $v$-token is thus the pair $($v-$\tau_{pub}, $v-$\tau_{priv})$ where v-$\tau_{pub} = (N, \sigma^i_{RA}(N))$ and v-$\tau_{priv} = ( \sigma^i_{RA}(N, o, w ) $, $\sigma^i_{RA}(N,o,p)$, $\sigma^i_{RA}(N,o,r))$. The number of $v$-tokens to produce initially can be easily deduced from the lowest higher bound in enforceable regulations
\footnote{For simplicity, we assume that there is at least one enforceable regulation in the system, and that the smallest threshold for all enforceable regulations is $\theta_{min}$. Then, $\theta_{min}$ is a sufficient upper bound of the time that can be spent by any given triplet of participants. %
It is therefore enough to produce $\theta_{min} \times |\mathcal{W}| \times |\mathcal{P}| \times |\mathcal{R}|$ $v$-tokens. In practice, the number of $v$-tokens produced can be drastically reduced in a straightforward manner by letting participants declare to the registration authority the subset of participants they may work with (\emph{e.g.,} selecting a subset of platforms or a few domains of interest).}

\subsubsection{The SPEND Function.}
\label{sec:spend}
Requesters create and send their tasks to a platform and the platform appends the tasks to either its own ledger (for local tasks) or the ledgers of all involved platforms (for cross-platform tasks). Once the task is published, the workers can indicate their intent to perform the task by sending a \emph{contribution intent} to their platforms. If a contribution is still needed for the task, the crowdworking process $\pi$ starts with the \texttt{SPEND} function, performed as follows. Without loss of generality, we assume below that the given task has a time cost equal to $1$.
First, for each enforceable regulation, the platform requests the public component of one $e$-token from a participant to the crowdworking process, whatever its group, i.e., worker, platform, or requester, that is also a target of the regulation. We call this participant {\em the initiator}.
For a verifiable regulation, the platform is the initiator and consequently does not request any $v$-token to the corresponding worker and requester.
For each regulation (each enforceable regulation and each verifiable regulation), the initiator then chooses one token, not spent yet, and sends it to the platform. Once the platform receives it, it forwards the full token to all the targets and requests two signatures (see below).
The platform includes in its request message the task (or a hash of it), an identifier for the contribution (e.g., a nonce, or when possible a hash of the contribution), and a signature of the two, concatenated, so that
the worker and the requester involved will be able to prove that they were asked for signatures, even if the platform fails.
The platform requires the worker and the requester to send back two signatures: %
(1) a group signature of the public component of the token and
(2) a group signature of the public component of the token concatenated to the task (or to its hash). %
The two signatures per participant will be used later for guaranteeing the correctness of the spending (see the CHECK function below). %
Finally, the platform append to the common datastore the public parts of the tokens ($\tau_{pub}$ for both $e$-tokens and $v$-tokens) and the signatures of the targets (($\sigma^g_w (\tau_{pub} || t), \sigma^g_r (\tau_{pub} || t), \sigma^g_p (\tau_{pub} || t)$) for both $e$-tokens and $v$-tokens).

\subsubsection{The PROVE function.}
\label{sec:prove}

Participants use the \texttt{PROVE} function to provide proofs for guaranteeing verifiable regulations (e.g., to a third party). The use of $v$-tokens is relatively straightforward. During the crowdworking processes, participants simply store the private components of $v$-tokens and deliver them at the end of the validity periods of verifiable regulations. As an example, for a $((w, p, r), >, 4)$ verifiable regulation, the worker $w$ sends the private components of $5$ $v$-tokens\footnote{ For minimal disclosure purposes, participants are allowed to send only the parts of the private components of $v$-tokens that are relevant to the verifiable regulations being verified: verifiable regulations that do not specify all the targets (e.g., $((w, *, *), >, \theta)$ only require the private components involving the worker $w$ (e.g., $\theta+1$ distinct $\sigma^i_{RA} (N , o, w))$ signatures where $o==w$ here).}. The entity in charge of guaranteeing the verifiable regulation checks the signature of the registration authority - to verify that the participant was involved in the task - and the nonce stored in the ledger - to ensure that the token has been indeed spent and committed to the ledgers of all platforms. %

\subsubsection{The CHECK and ALERT Functions.}
\label{sec:check}

These functions are used to detect and report either the malicious behavior of participants resulting in an invalid consumption of tokens or the failure of a platform. The complete set of verifications protects against
(1) the forgery of tokens (verification of the signatures),
(2) the replay of tokens (verification of the absence of double-spending),
(3) the relay of tokens (verification of the absence of usurpation), and
(4) the illegitimate invalidation of tokens (timeout against malicious platform failures).
The first two verifications are straightforward and performed during global consensus. 
Verifications (3) and (4) are similar, and we explain here case (3).
When a token is appended to the ledgers of all platforms, any participant (whether involved in the corresponding crowdworking process or not) can \texttt{CHECK} its nonce.
If a participant detects a nonce that was received from the registration authority but not spent (or spent on a wrong task),
she \texttt{ALERT}s the registration authority. The registration authority will de-anonymize the group signature of the corresponding participant (e.g., the worker's group signature if the alert comes from a worker) and check whether it has been signed by the same participant that sent the token.
Depending on the result, the registration authority will take adequate sanctions against the fraudulent participant (true positive) or the alert-riser (false positive).

\medskip
\noindent
{\bf Platform failure.}
If a platform fails after it requested $e$-tokens or signatures and does not recover (e.g., tokens are not appended to the ledger), the $e$-tokens
revealed to the platform are lost: they cannot be used in any other crowdworking process because they are not anonymous anymore (i.e., the platform knows the association between them and the corresponding participants), and they are not spent either.  In that case, workers or requesters send an \texttt{ALERT} to the registration authority including
(1) the public key of the platform,
(2) the task (or a hash of it), and
(3) all the requests received from the platform.
The registration authority then checks whether the number of requests sent by the platform for the given task matches the corresponding number of matching messages committed in the ledger of all participants.
If there are more requests, the registration authority sets a timeout (e.g., to let unfinished transactions end or the platform recover from a failure).

\begin{figure}[t]
\begin{center}
\includegraphics[width= 0.7\linewidth]{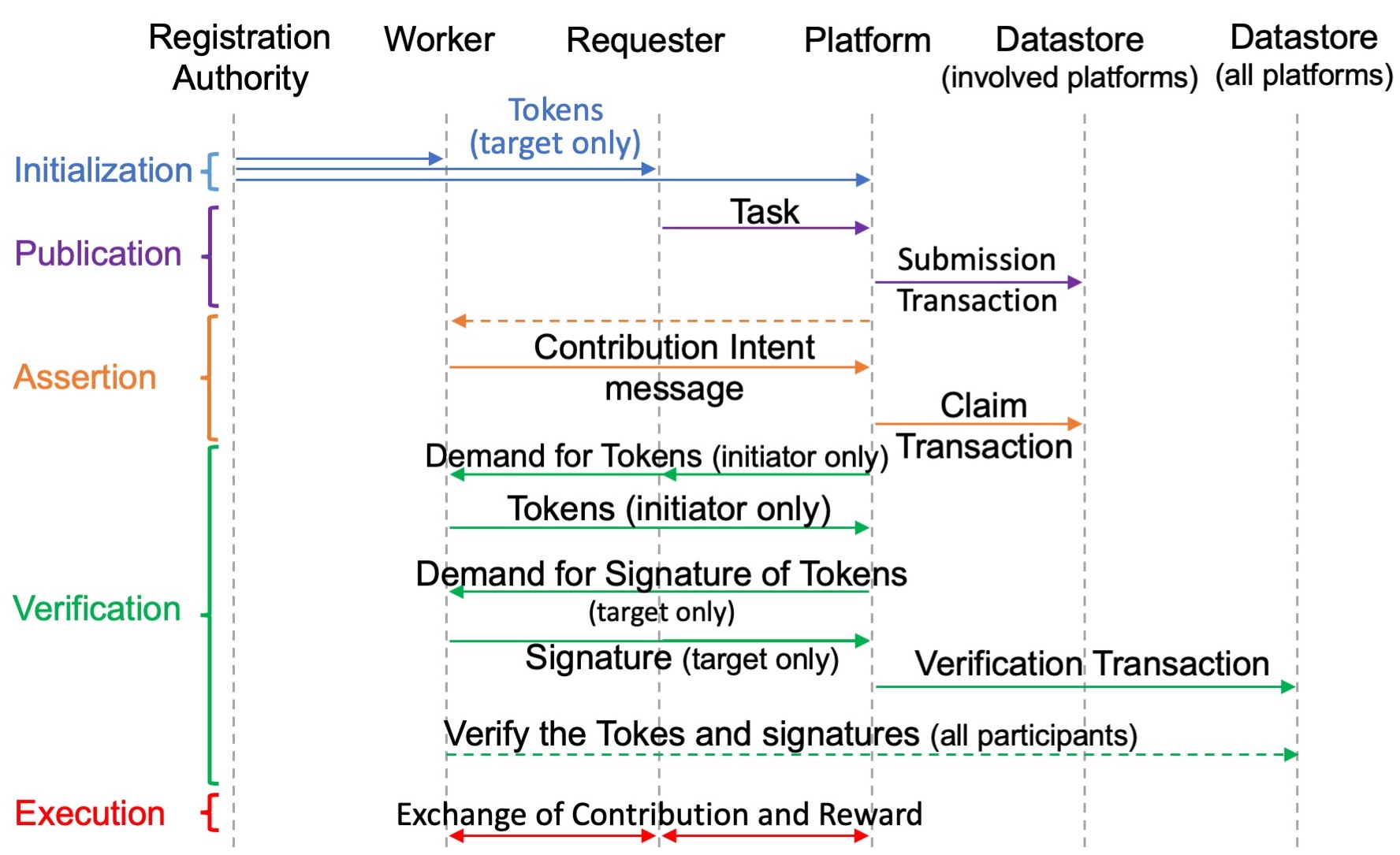}
\caption{Sequence chart {\small (references to targets include all participants for $v$-tokens)}}
\label{fig:constraints}
\end{center}
\end{figure}

\subsection{Task Processing Sequence}\label{sec:global_exec}

The processing of a crowdworking task involves the following five main phases, as depicted in Figure~\ref{fig:constraints}.

\noindent
{\bf Initialization.} The registration authority provides all parties with their keys and tokens. 

\noindent
{\bf Publication.}
Requesters create and send their tasks to platforms.
If a requester wants to publish its task on more than one platform (i.e., a cross-platform task),
the involved platforms collaborate with each other to create a common instance of the task
(e.g., a common task identifier).
The involved platforms then append the task to their ledgers through \sub transactions and
inform their workers in their preferred manner for accessing tasks.

\noindent
{\bf Assertion.}
After a worker has retrieved a task, the worker sends a {\em contribution intent} message
to the platform without revealing the actual contribution.
The platform then updates the number of required contributions for the task and
appends the contribution intent to its ledger through a \clm transaction.
For cross-platform tasks, the platform informs other involved platforms
about the received contribution intent, so that all involved platforms
agree with the number (and order) of the received contribution intents
(i.e., \clm transactions) and append the \clm transaction to their ledgers.
If the desired number of contribution for the task has been achieved, the contribution intent is refused.

Note that while the requester does not choose the workers, it is possible to enforce a selection with {\em a priori} criteria, passed through the platform. Another straightforward enhancement would be to add a communication step, by forwarding the contribution intent, together with the worker's identity to the requester, and letting her approve of it or not. This communication, however, requires the disclosure of the worker's identity to requesters even before a contribution is accepted.

\noindent
{\bf Verification.} Once the contribution intent has been accepted by the platform(s),
the platform asks the corresponding requester and worker to send the required tokens and signatures,
through the \texttt{SPEND} function (see above).
Upon receiving all tokens and signatures, the platform shares them with all platforms and
the $e$-tokens and their signatures are appended to the ledgers of all platforms through \ver transactions.
From this point, anyone can check the validity of requirements
with the \texttt{CHECK} function (and \texttt{ALERT} if required), as developed above.

\noindent
{\bf Execution.} Once all parties have checked the validity of the task, the tokens, and the group signatures, the contribution can be delivered to the requester and the reward to the worker. %

\subsection{Privacy Analysis}

We show below in the suite of Theorem~\ref{the:conf}, Lemma~\ref{lem:integr}, Lemma~\ref{lem:integr_cert}, and Theorem~\ref{the:conf-cov} that the global execution of \system satisfies the \emph{$\delta^\Pi$-privacy} model against covert adversaries (where the disclosure sets are defined in Section~\ref{sec:model}).

First, Theorem~\ref{the:conf} restricts the adversarial behavior to inferences (i.e., similar to a \emph{honest-but-curious} adversary) and shows that the execution of \system satisfies $\delta^\Pi$-privacy (where the disclosure sets are defined in Section~\ref{sec:model}). 

\begin{theorem}\label{the:conf}(\textit{Privacy (inferences)})
  For all sets of crowdworking processes $\Pi$ executed over
  participants $\mathcal{W}$, $\mathcal{P}$, and $\mathcal{R}$ by an
  instance of \system $\varsigma$, then it holds that $\varsigma$ is
  $\delta^\Pi$-Private against covert adversaries restricted to inferences (where the disclosure sets are defined in Section~\ref{sec:model}).
\end{theorem}

\ifextend
\begin{proof}(Sketch)
First, we focus on the content of tokens and show that it is harmless. For each crowdworking process $\pi$, the information contained within the tokens exchanged and stored in the common datastore is made of (1) the public components $e$-$\tau_{pub}$ and $v$-$\tau_{pub}$ of the tokens involved (\emph{i.e.,} a nonce and a signature) and (2) of the group signatures of the participants to $\pi$. The nonce is generated by the registration authority independently from $\pi$, thus do not leak any information about $\pi$\footnote{Including a generation timestamp into the public component of tokens, for supporting validity periods, would not leak information about $\pi$ beyond its probable execution timeframe, which is already captured by the $BEGIN$ and $END$ events allowed to leak in the $\delta^\Pi$-privacy model. Indeed, the timestamps would be generated by the registration authority independently from $\pi$.}.
Since the group signatures are generated by a semantically-secure group signature scheme, they do not leak anything to the real adversary (computationally bounded) beyond the groups of the signers, which is also available to the ideal adversary. %

Second, we concentrate on the information disclosed along the execution sequence of crowdworking processes. We consider below each disclosure set $\delta_i^\pi$ in turn and show that the computational indistinguishability requirement between the ideal setting and the real setting (where the instance of \system $\varsigma$ executes $\Pi$) is satisfied in all cases. %
\emph{Disclosure set $\delta_{\neg R \neg I}^\pi$.} %
The $\delta_{\neg R \neg I}^\pi$ disclosure set contains the information allowed to be disclosed to the participants that are not involved in a crowdworking process $\pi \in \Pi$ and that have not received the related task (\emph{i.e.,} $(\mathtt{BEGIN}, \mathtt{END}, p)$). First, we focus on the subset of such participants that are requesters or workers. In the ideal settings, these participants learn nothing beyond the information contained in $\delta_{\neg R \neg I}^\pi$. In the real setting, when $\pi$ is executed by $\varsigma$, these participants are not involved in any consensus. They are only able to observe the state of the datastore. The latter is updated exactly once for $\pi$, when $\pi$ ends, for storing the tokens spent: only the ending event $(\mathtt{END})$ is disclosed, which is already contained within $\delta_{\neg R \neg I}^\pi$. Second, we focus on the platforms. In the ideal setting, they are given $\delta_{\neg R \neg I}^\pi$. In the real setting, they participate to the global consensus and are able to observe the state of the datastore. Consequently, they learn $p$ from the global consensus (i.e., the platform that initiates the global consensus) and the ending event $\mathtt{END}$\footnote{ The verifications performed by the worker and the requester involved in $\pi$ are performed on the instance of the datastore stored on $p$. Indeed, the block resulting from the global consensus contains all the necessary information both for performing the verification and for checking that it results from the global consensus.}. Both are already contained within $\delta_{\neg R \neg I}^\pi$. %

\emph{Disclosure set $\delta_{R\neg I}^\pi$.} The $\delta_{R\neg I}^\pi$ disclosure set contains the information allowed to be disclosed to the workers and platforms that have received the task $t$ from $r$ but that are not involved in the crowdworking process $\pi \in \Pi$ (\emph{i.e.,} $(\mathtt{BEGIN},$ $\texttt{END}, p, r, t)$). First, we focus on the subset of such participants that are workers. In the ideal settings, they learn nothing beyond the information contained in $\delta_{R\neg I}^\pi$. In the real setting, we see from the global execution sequence of \system that (1) they are not involved in any consensus, (2) but are able to observe the state of the common datastore, (3) have received the task $t$, and (4) may receive an abort from their platform if they contribute to $t$ while $t$ has already been solved.  From (2) and (4), they are able to learn the ending events of crowdworking processes, from (3) they learn $t$, and from (1) they do not learn any other information. As a result, they learn $(\mathtt{END}, t)$, which is contained in $\delta_{R\neg I}^\pi$. Second, we focus on the subset of participants that are platforms. In the ideal settings, they learn nothing beyond the information contained in $\delta_{R\neg I}^\pi$. In the real setting, (1) they receive the tasks from the requesters, (2) they participate to the cross-platform consensuses and to the global consensuses in addition to (3) observing the same information as the workers. From (1) they learn $t$ and $r$ for each process $\pi \in \Pi$. From (2), they learn $p$, \texttt{BEGIN}, and \texttt{END} because: the cross-platform consensus discloses the initiating platform and the starting event, and the global consensus discloses the initiating platform together with the ending event. From (3), they learn the ending event and the task. As a result, such platforms learn about all $\pi \in \Pi$ the following information $(\mathtt{BEGIN}, \mathtt{END}, r, p, t)$, which is exactly $\delta_{R\neg I}^\pi$. %

\emph{Disclosure set $\delta_{RI}^\pi$.} The computational indistinguishability requirement between the real setting and the ideal setting is trivially satisfied for the $\delta_{RI}^\pi$ set of disclosures because $\delta_{RI}^\pi$ contains all the information about the execution of the crowdworking process $\pi$ so $\varsigma$ does not (and cannot) disclose more information. %

\end{proof}

Second, we extend possible behavior to malicious behaviors aiming at jeopardizing regulations
and show that they are systematically detected by \system (Lemma~\ref{lem:integr} focuses on enforceable regulations and Lemma~\ref{lem:integr_cert} on verifiable regulations). This prevents covert adversaries from performing malicious actions, limiting them to inferences.

\begin{lemma}\label{lem:integr}(\textit{Detection of malicious behavior (enforceable regulations)})
A crowdworking process $\pi$ executed over participants $\mathcal{W}$, $\mathcal{P}$, and $\mathcal{R}$ by an instance of \system $\varsigma$, completes successfully without raising a legitimate alert if and only if $\pi$ does not jeopardize any enforceable regulation.
\end{lemma}

\ifextend
\begin{proof}(Sketch)
We have to show first that the $e$-tokens allocated to participants can be spent, and second that participants cannot spend more. %

\textbf{Participants can spend their $e$-tokens.} %
In order to prevent participants from spending their $e$-tokens, an attacker has three main possibilities. First, she can try to acquire tokens belonging to another participant and to spend them (i.e., relay attack). However, this rises an \texttt{ALERT} with certainty. Indeed, if a token is spent by an illegitimate participant, it is stored on the datastore and is thus accessible to the legitimate participant who is able to detect it through the nonce and to rise an \texttt{ALERT} to the registration authority (including the group signatures stored along the token). Second, the attacker (platform only) could try to misuse tokens by spending them in a way that was not intended by the legitimate owner (i.e., relay attack). However, this would be detected by participants as well because the signature of the task would not be valid. Third, the attacker (platform only) may abort the process after having received tokens but before performing the global consensus (i.e., illegitimate invalidation). However, after a timeout, the other involved participants simply send an \texttt{ALERT} to the registration authority and prove that their tokens were requested by the platform (signatures of the requests for tokens and signatures), and therefore that the platform behaves illegitimately.

\textbf{Participants cannot spend more.} %
First, an attacker may produce additional $e$-tokens (i.e., forge attack). 
However, the public parts of $e$-tokens must contain valid signatures produced by the registration authority.
Second, an attacker may try to spend an $e$-token more than once (i.e., replay attack). However, the nonce of an $e$-token that must be spent must not already be in the datastore. %
Finally, an attacker may simply omit sending any $e$-token. However, the public parts of $e$-tokens are required for the successful completion of the global consensus. %
\end{proof}

\begin{lemma}\label{lem:integr_cert}(\textit{Detection of malicious behavior (verifiable regulations)})
Participant $P$ can produce a proof about process $\pi$ executed over participants $\mathcal{W}$, $\mathcal{P}$, and $\mathcal{R}$ by an instance of \system $\varsigma$ if and only if $P$ was involved in $\pi$ {\bf and}  $\pi$ completed successfully. 
\end{lemma}

\begin{proof}(Sketch)
First, a proof about a crowdworking process $\pi$ that completed successfully can always be produced by the participants involved in $\pi$. Indeed, the $v$-tokens are produced by the registration authority and sent to all participants (i.e., the number of $v$-tokens is correct), and only the successful crowdworking processes store the $v$-tokens in the datastore. %
Second, participants cannot spend a $v$-token more than once because the nonce of a token that must be spent must not already be in the common datastore. Third, participants cannot produce any $v$-token by themselves because their public parts must contain valid signatures produced by the RA.
\end{proof}

Since Theorem~\ref{the:conf} shows that \system is $\delta^\Pi$-private against adversaries restricted to inferences, and Lemma~\ref{lem:integr} and Lemma~\ref{lem:integr_cert} show that malicious behaviors are prevented, it follows that \system is $\delta^\Pi$-private against covert adversaries (Theorem~\ref{the:conf-cov}). 

\begin{theorem}\label{the:conf-cov}(\textit{Privacy (inferences and malicious behavior)})
For all sets of crowdworking processes $\Pi$ executed over participants $\mathcal{W}$, $\mathcal{P}$, and $\mathcal{R}$ by an instance of \system $\varsigma$, then 
$\varsigma$ is $\delta^\Pi$-private against covert adversaries (where the disclosure sets are defined in Section~\ref{sec:model}).
\end{theorem}

\begin{proof}
Theorem~\ref{the:conf} shows that the execution of
\system satisfies $\delta^\Pi$-privacy against covert adversaries restricted to inferences, while Lemmas~\ref{lem:integr} and \ref{lem:integr_cert} show that malicious behaviors aiming at jeopardizing the regulations guaranteed are detected and consequently prevented within \system. As a result it follows directly that the execution of \system satisfies $\delta^\Pi$-privacy against covert adversaries (where the disclosure sets are defined in Section~\ref{sec:model}).
\end{proof}

\section{Coping with Distribution}
\label{sec:dist}

Section~\ref{sec:enf} provides an abstract design for implementing regulations.
In this section, we show how \system,
supports the execution of transactions
on multiple globally distributed platforms that do not necessarily trust each other.
In \system and in order to provide fault tolerance,
each platform consists of
a set of nodes (i.e., replicas) that store copies of the platform's ledger.
\system uses a {\em permissioned blockchain} as its underlying infrastructure (i.e., MPI).
The unique features of blockchain such as
transparency, provenance, fault tolerance, and authenticity are used by many systems
to deploy a wide range of distributed applications in a permissioned settings.
In particular and for a crowdworking system,
the {\em transparency} of blockchains is useful for checking integrity constraints,
{\em provenance} enables \system to trace how data is transformed,
{\em fault tolerance} helps to enhance reliability and availability, and finally,
{\em authenticity} guarantees that signatures and transactions are valid.

\subsection{Blockchain Ledger}

\begin{figure}[t] \center
\includegraphics[width=0.5\linewidth]{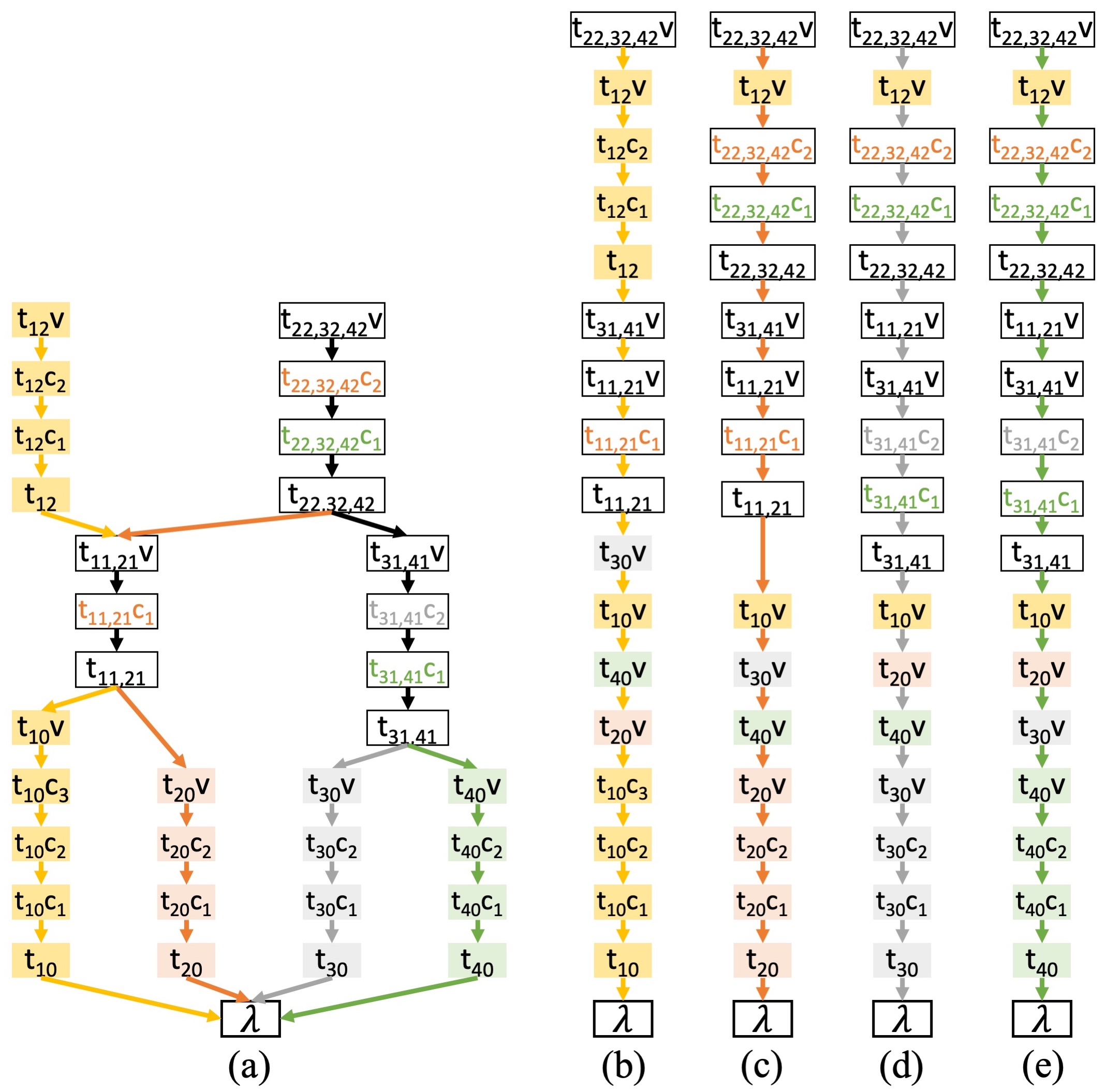}
\caption{(a): The ledger of \system with $4$ platforms,
(b), (c), (d), and (e): The views of the ledger from different platforms}
\label{fig:ledger}
\end{figure}

The blockchain ledger in \system, as mentioned before, maintains the global state of the system and
includes all \sub, \clm, and \ver transactions of
all internal as well as cross-platform tasks.
To ensure data consistency,
an ordering among transactions in which a platform is involved is needed.
The order of transactions in the blockchain ledger is captured by
{\em chaining} transaction blocks together, i.e., 
each transaction block includes the sequence number or the cryptographic hash of the previous transaction block.
Since \system supports both internal and cross-platform tasks and
more than one platform are involved in each cross-platform transaction,
the ledger (similar to \cite{amiri2019caper, amiri2019sharper}) is formed as a {\em directed acyclic graph (DAG)}
where the {\em nodes} of the graph are transaction blocks (each block includes a single transaction)
and {\em edges} enforce the order among transaction blocks.
In addition to \sub, \clm, and \ver transactions, a unique initialization transaction (block),
called the {\em genesis} transaction is also included in the ledger.

Fig.~\ref{fig:ledger}(a) shows a blockchain ledger created in the \system model
for a blockchain infrastructure
consisting of four platforms $p_1$, $p_2$, $p_3$, and $p_4$.
In this figure, $\lambda$ is the unique initialization ({\em genesis}) block of the blockchain,
$t_i$'s are \sub transactions,
$t_ic_j$ is the $j$-th \clm transaction of task $t_i$, and
$t_iv$ is the \ver transaction of task $t_i$.
In Fig.~\ref{fig:ledger}(a), $t_{10}$, $t_{20}$, $t_{30}$, and $t_{40}$ are internal \sub transactions
of different platforms that can be appended to the ledger in parallel.
As shown, $t_{10}$ requires $3$ contributions
(thus $3$ \clm transactions $t_{10}c_1$, $t_{10}c_2$, and $t_{10}c_3$)
whereas each of $t_{20}$, $t_{30}$, and $t_{40}$ needs two contributions.
$t_{10}v$, $t_{20}v$, $t_{30}v$, and $t_{40}v$ are the corresponding  \ver transactions.
$t_{11,21}$ is a cross-platform \sub among platforms $p_1$ and $p_2$.
Similarly, $t_{31,41}$ is a cross-platform \sub among platforms $p_3$ and $p_4$.
Here, $t_{11,21}$ and $t_{31,41}$ require one and two contributions respectively.
Note that the \clm transactions of a cross-platform task might be initiated by different platforms
and as mentioned earlier, the order of these \clm transactions is important (to recognize the $n$ first \clms).

This global directed acyclic graph blockchain ledger includes all transactions of
internal as well as cross-platform tasks initiated by all platforms.
However, to ensure data privacy,
each platform should only access the transactions in which the platform is involved.
As a result, in \system,
the entire blockchain ledger is {\em not maintained} by any specific platform and
each platform $p_i$, as shown in Fig.~\ref{fig:ledger}(b)-(e),
only maintains its own {\em view} of the blockchain ledger including
(1) all \sub and \clm transactions of its internal tasks,
(2) all \sub and \clm transactions of the cross-platform tasks involving the platform, and
(3) \ver transactions of all tasks.
Note that \ver transactions are replicated on every platform
to enable all platforms to check the satisfaction of global regulations.
The global DAG ledger (e.g., Fig.~\ref{fig:ledger}(a)) is
indeed the union of all these physical views
(e.g., Fig.~\ref{fig:ledger}(b)-(e)).
Note that, since there is no data dependency between 
the tasks that platform $p$ is involved in and
the \ver transactions of the tasks that platform $p$ is {\em not} involved in,
the \ver transactions might be appended to the ledgers in different orders, e.g.,
$t_{20}v$ (of $p_2$) and $t_{40}v$ (of $p_4$) are
appended to the ledger of platforms $p_1$ and $p_3$ in different orders.

\subsection{Consensus in \system}

\begin{figure}[t]
\begin{center}
\includegraphics[width= 0.4\linewidth]{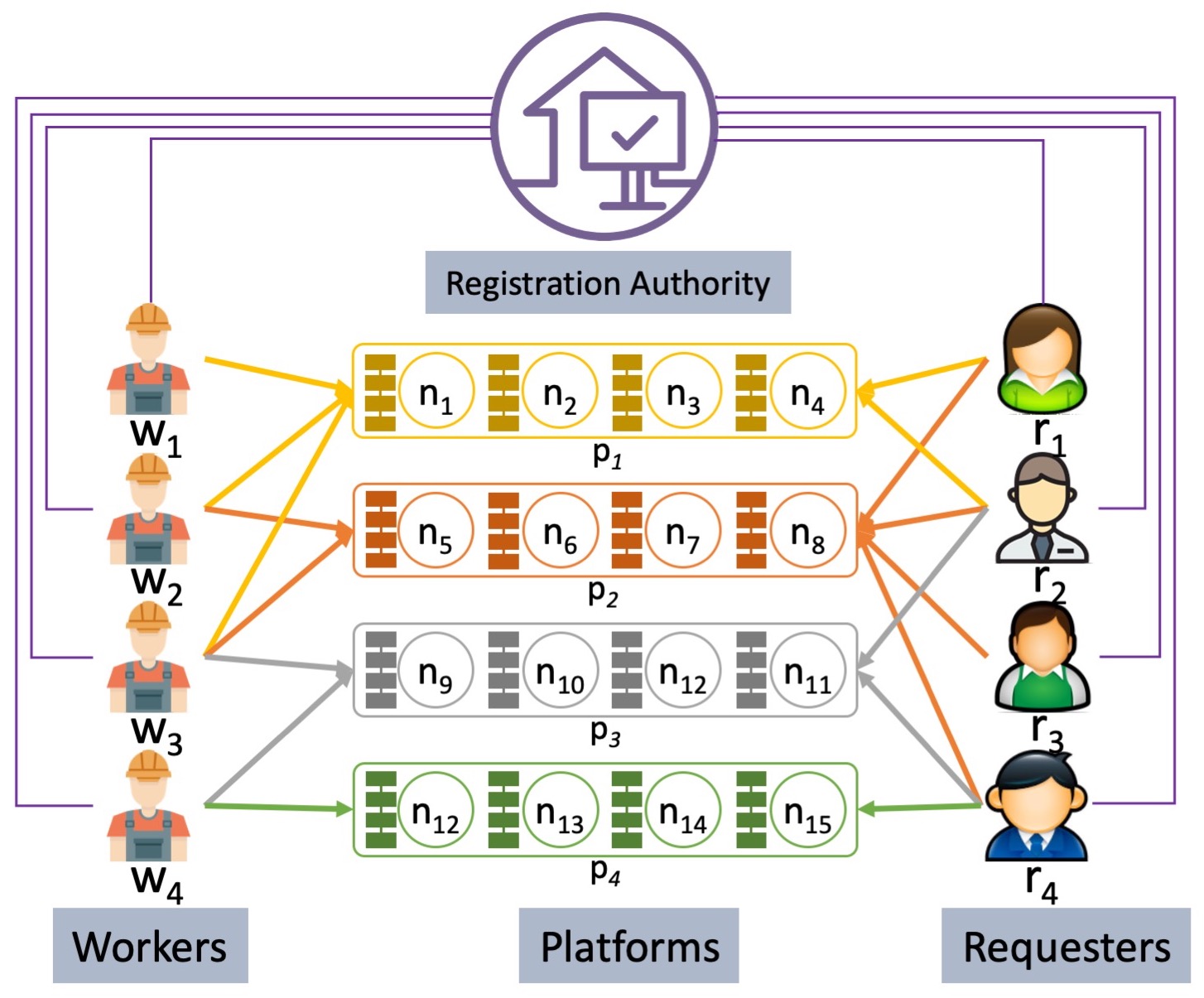}
\caption{\system infrastructure}
\label{fig:arch-separ}
\end{center}
\end{figure}

In \system, each platform consists of a (disjoint) set of nodes (i.e., replicas) where
the platform replicates its own view of the blockchain ledger
on those nodes to achieve fault tolerance.
Nodes follow either the crash or Byzantine failure model.
In the crash failure model,
nodes may fail by stopping and may restart, however,
in the Byzantine failure model,
faulty nodes may exhibit malicious behavior.
Nodes of the same or different platforms need to establish consensus on a unique
order in which entries are appended to the blockchain ledger.
To establish consensus among the nodes, asynchronous fault-tolerant protocols have been used.
Crash fault-tolerant protocols, e.g., Paxos \cite{lamport2001paxos},
guarantee safety in an asynchronous network using $2f{+}1$ nodes
to overcome the simultaneous failure of any $f$ nodes while
in Byzantine fault-tolerant protocols, e.g., PBFT \cite{castro1999practical},
$3f{+}1$ nodes are usually needed to provide safety in the presence of $f$ malicious nodes.
Figure~\ref{fig:arch-separ} shows the crowdworking infrastructure of Figure~\ref{fig:arch} where
each platform consists of 4 replicas (assuming Byzantine failure model and $f=1$)
and replicas use a blockchain to store data.

Completion of a crowdworking task, as discussed earlier, requires
a single \sub transaction,
one or more \clm transactions (depending on the requested number of contributions), and
a \ver transaction.
For an internal task of a platform, \sub and \clm transactions
are replicated only on the nodes of the platform,
hence, {\em local consensus} among nodes of the platform on the order of the transaction is needed.
For a cross-platform task, on the other hand,
\sub and \clm transactions are replicated on every node of all involved platforms.
As a result, {\em cross-platform consensus} among the nodes of only the {\em involved} platforms is needed.
Finally, \ver transactions are appended to the ledger of all platforms, therefore,
all nodes of {\em every} platform participate in a {\em global consensus} protocol.
In this section, we show how local, cross-platform, and global consensus are established
with crash-only or Byzantine nodes.

\subsubsection{Local Consensus}
Processing a \sub or a \clm transaction of an internal task requires
local consensus where
nodes of a single platform, {\em independent} of other platforms,
establish agreement on the order of the transaction.
The local consensus protocol in \system is pluggable and
depending on the failure model of nodes, i.e., crash-only or Byzantine,
a platform uses a crash fault-tolerant protocol, e.g., Paxos \cite{lamport2001paxos}, or
a Byzantine fault-tolerant protocol, e.g., PBFT \cite{castro2002practical}.
Figure~\ref{fig:crash} shows the normal case operation of both Paxos and PBFT protocols.

\begin{figure}[h] \center
\begin{minipage}{.31\textwidth}\centering
\includegraphics[width= \linewidth]{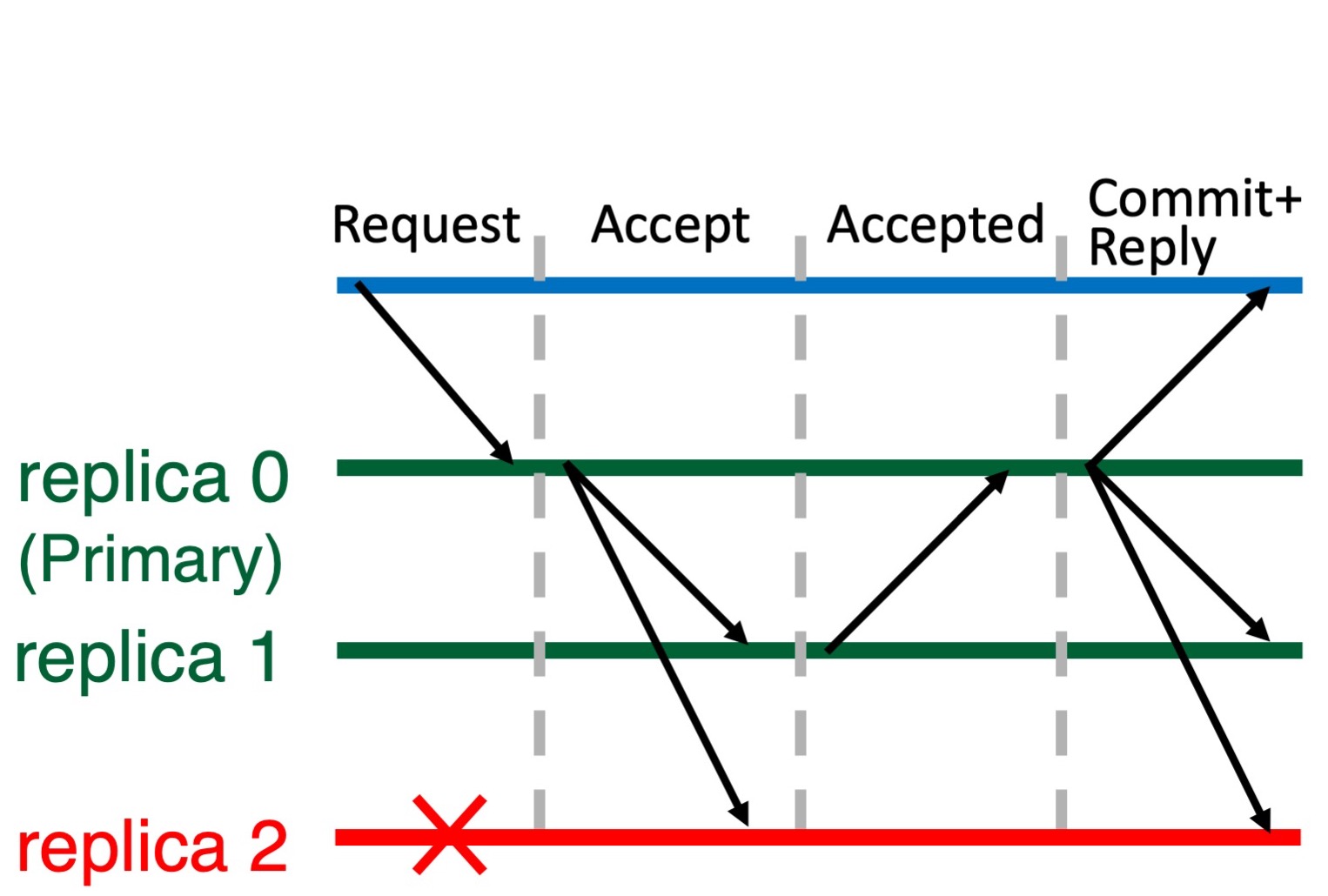}
\end{minipage}
\begin{minipage}{.382\textwidth}\centering
\includegraphics[width= \linewidth]{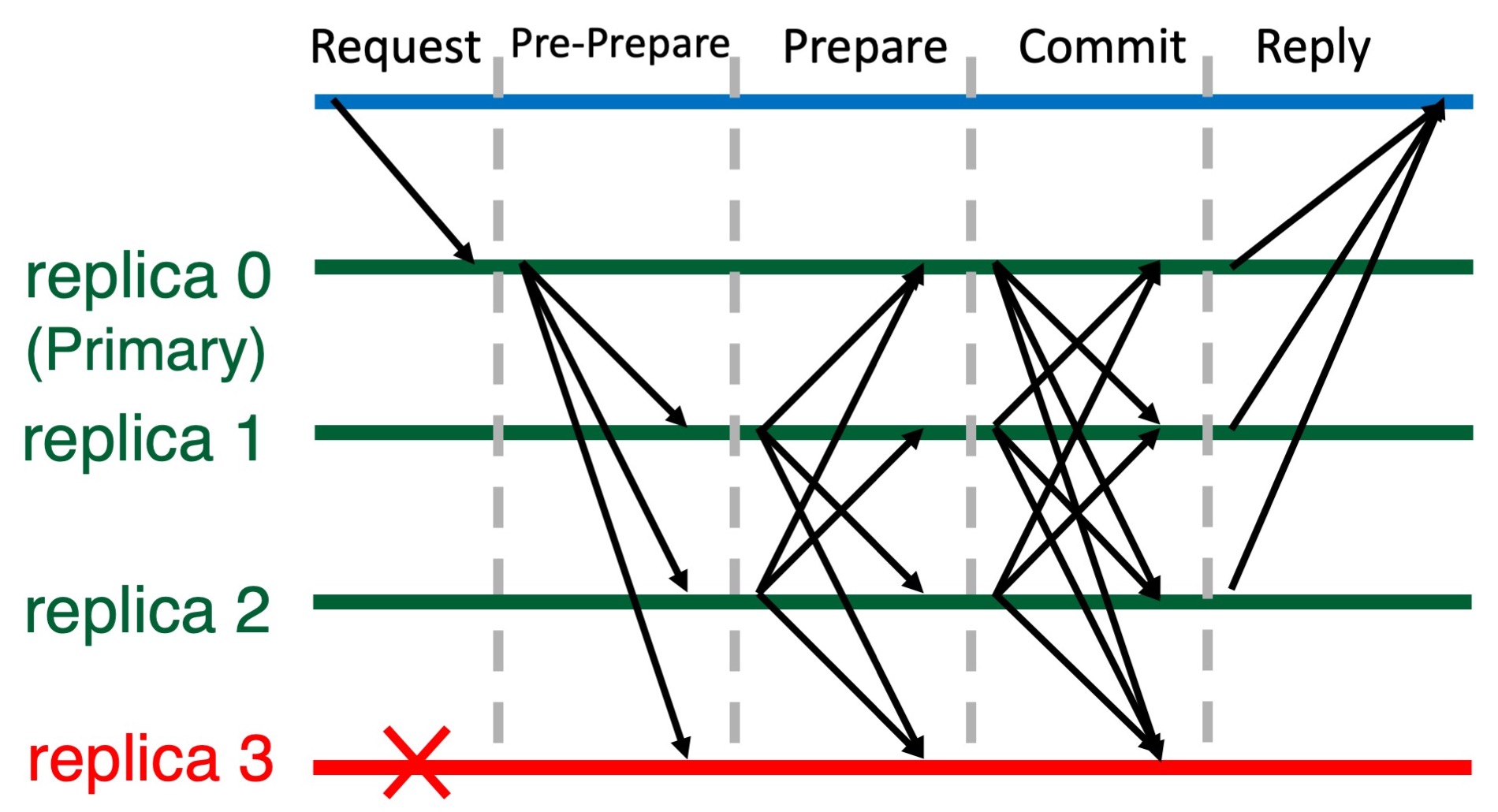}
\end{minipage}
\caption{Normal case operation in (a) Paxos \cite{lamport2001paxos} and (b) PBFT \cite{castro1999practical}}
\label{fig:crash}
\end{figure}

The local consensus protocol
is initiated by a pre-elected node of the platform, called {\em the primary}.
When the primary $p$ receives a valid internal transaction (either \sub or \clm),
it initiates a local consensus algorithm
by multicasting a message, e.g.,
{\sf accept} message in Paxos or {\sf pre-prepare} message in PBFT,
including the requested transaction to other nodes of the platform.
To provide a total order among transaction blocks, the primary also assigns a sequence number to the request.
Instead of a sequence number, the primary can also include
the cryptographic hash of the previous transaction block in the message.
If the transaction is a \clm transaction,
the primary includes the cryptographic hash of the corresponding \sub transaction and
any previously received \clm transactions for that particular task (if any).
The nodes of the platform then establish agreement on a total order of
transactions using the utilized consensus protocol and
append the transaction to the blockchain ledger.

\subsubsection{Cross-Platform and Global Consensus}

Both cross-platform consensus and global consensus require collaboration between multiple platforms.
Since platforms do not trust each other and
the primary node that initiates the transaction might behave maliciously,
\system uses {\em Byzantine} fault-tolerant protocols
in both cross-platform and global consensus.
cross-platform and global consensus, however, are different in two aspects.
First, in cross-platform consensus, only the involved platforms participate, whereas
global consensus is established among all platforms, and
second, at the platform level,
while cross-platform consensus requires agreement from {\em every} involved platform,
in global consensus, agreement from {\em two-thirds} of all platforms is sufficient.
Cross-platform consensus requires agreement from
every involved platform to ensure data consistency
due to the possible data dependency between the cross-platform transaction
and other transactions of an involved platform.
Note that if an involved platform (as a set of nodes) behaves maliciously by not sending an agreement
for a cross-platform transaction initiated by another platform, e.g., a \clm transaction,
its malicious behavior can be detected and penalties imposed.
In global consensus, however, the goal is only to check the correctness of the transaction.
To provide safety for global consensus at the platform level, 
we assume that at most $\lfloor \frac{|P|-1}{3} \rfloor$ platforms might behave maliciously.
As a result, to commit a transaction,
by a similar argument as in PBFT \cite{castro1999practical},
at least two-thirds ($\lfloor \frac{2|P|}{3} \rfloor+1$) of the platforms
must agree on the order of the transaction.

\noindent
{\bf Cross-Platform Consensus.}
Processing \Sub and \clm transactions of a cross-platform task
requires cross-platform consensus among {\em all} involved platforms where
due to the untrustworthiness of platforms a Byzantine fault-tolerant protocol is used.
Since the number of nodes within each platform depends on the failure model of nodes of a platform
(i.e. $2f+1$ crash-only or $3f+1$ Byzantine nodes),
the required number of matching replies from each platform, i.e., the quorum size,
to ensure the safety of protocol is different for different platforms.
We define {\em local-majority} as
the required number of matching replies from the nodes of a platform.
For a platform with crash-only nodes,
local-majority is $f+1$ (from the total $2f+1$ nodes), whereas
for a platform with Byzantine nodes,
local-majority is $2f+1$ (from the total $3f+1$ nodes).

\system processes cross-platform transactions in four phases: \zero, \one, \two, and \three.

Upon receiving a cross-platform (\sub or \clm) transaction,
the (pre-elected) primary node of the (recipient) platform initiates the consensus protocol
by multicasting a \zero message to the primary node of all involved platforms.
Each primary node then assigns a sequence number to the request and multicasts
a \one message to every node of its platform.
During the \two and \three phases,
all nodes of every involved platform communicate to each other to reach agreement
on the order of the cross-platform transaction.

\newcommand{\azero}{{\tiny $\langle\langle\text{\ZERO}{,} h_i{,} d \rangle_{\sigma_{\pi(p_i)}}, m \rangle$}\xspace}
\newcommand{\azeror}{{\tiny $\langle\langle\text{\ZERO}, h_i, d \rangle_{\sigma_r}, m \rangle$}\xspace}
\newcommand{\aone}{{\tiny $\langle\langle\text{\ONE}, h_i, d \rangle_{\sigma_{\pi(p_i)}}, m \rangle$}\xspace}
\newcommand{\aoner}{{\tiny $\langle\langle\text{\ONE}, h_i, d \rangle_{\sigma_r}, m \rangle$}\xspace}
\newcommand{\aonej}{{\tiny $\langle\langle\text{\ONE}, h_j, d, r \rangle_{\sigma_{\pi(p_j)}}, \mu \rangle$}\xspace}
\newcommand{\aonek}{{\tiny $\langle\langle\text{\ONE}, h_j, d, r \rangle_{\sigma_r}, \mu \rangle$}\xspace}
\newcommand{\atwo}{{\tiny $\langle\text{\TWO}, h_i, h_j, d, r \rangle_{\sigma_r}$}\xspace}
\newcommand{\atwoi}{{\tiny $\langle\text{\TWO}, h_i, d, r \rangle_{\sigma_r}$}\xspace}
\newcommand{\atwoj}{{\tiny $\langle\text{\TWO}, h_i, h_j, d, r \rangle_{\sigma_{\pi(p_j)}}$}\xspace}
\newcommand{\athree}{{\tiny $\langle\text{\THREE}, h_i, h_j, ..., h_k, d, r \rangle_{\sigma_r}$}\xspace}

\begin{algorithm}[t]
\scriptsize
\caption{{\small Cross-Platform Consensus}}
\label{alg:cross}
\begin{algorithmic}[1]
\State {\bf init():} 
\State \quad $r$ := {\em node\_id}
\State \quad $p_i$ := the platform that initiates the consensus
\State \quad $\pi(p)$ := the primary node of platform $p$
\State \quad $P$ := the set of involved platforms
\State \quad $\pi(P)$ := the primary nodes of platforms in $P$
\newline
\State {\bf upon receiving} valid transaction $m$ and $(r == \pi(p_i))$
\State \quad {\bf multicast} \azero to $\pi(P)$
\State \quad {\bf multicast} \aone to all nodes of $p_i$ 
\newline
\State {\bf upon receiving} valid $\mu{=}$ \azero and $r {==} \pi(p_j)$
\State \quad if $r$ is not involved in any uncommitted request $m'$
where $m$ and $m'$ intersect in some other platform $p_k$
\State \qquad {\bf multicast} \aonej to all nodes of $p_j$
\State \qquad {\bf multicast} \atwoj to $P$
\newline
\State {\bf upon receiving} valid \aone and $r \in p_i$
\State \qquad {\bf multicast} \atwoi to $P$
\newline
\State {\bf upon receiving} valid \aonej and $r \in p_j$
\State \qquad {\bf multicast} \atwo to $P$
\newline
\State {\bf upon receiving} valid matching \atwo from {\em local-majority} of every platform $p_j$ in $P$
\State \quad {\bf multicast} \athree to $P$
\newline
\State {\bf upon receiving} valid \athree from {\em local-majority} of every platform in $P$
\State \quad {\bf append} the transaction block to the ledger
\end{algorithmic}
\end{algorithm}

Algorithm~\ref{alg:cross} presents the normal case of {\em cross-platform consensus} in \system.
Although not explicitly mentioned, every sent and received message is logged by nodes.
As shown in lines 1-6 of the algorithm,
$p_i$ is the platform that initiates the transaction,
$\pi(p)$ represents the primary node of platform $p$,
$P$ is the set of involved platforms in the transaction where
$\pi(P)$ represents their current primary nodes (one node per platform).

Once the primary $\pi(p_i)$ of the initiator platform $p_i$ (called the {\em initiator primary})
receives a valid \sub or \clm transaction,
as presented in lines 7-8,
the initiator primary node
assigns sequence number $h_i$ to the request and
multicasts a {\em signed} \zero message
$\langle\langle\text{\scriptsize \ZERO}, h_i, d \rangle_{\sigma_{\pi(p_i)}}, m \rangle$
to the primary nodes of all involved platforms where
$m$ is the received message (either \sub or \clm) and
$d=D(m)$ is the digest of $m$.
The sequence number $h_i$ represents the correct order of the transaction block in the initiator platform $p_i$.
If the transaction is a \clm transaction,
the primary includes the cryptographic hash of the corresponding \sub transaction as well.

As shown in line 9, the initiator primary node also multicasts a {\em signed} \one message
$\langle\langle\text{\scriptsize \ONE}, h_i, d \rangle_{\sigma_{\pi(p_i)}}, m \rangle$
to the nodes of its platform where $d=D(m)$ is the digest of $m$.

As indicated in lines 10-12, once the primary node of some platform $p_j$ receives a \zero message $\mu$
from the initiator primary,
it first validates the message.
If the node is currently waiting for a \three message of some cross-platform transaction $m'$ where
the involved platforms of the two requests $m$ and $m'$ intersect in $p_j$ as well as some other platform $p_k$,
the node does not process the new transaction $m$ before the earlier transaction $m'$ gets committed.
This ensures that requests are committed in the same order on overlapping platforms (Consistency), e.g.,
$m$ and $m'$ are committed in the same order on both $p_j$ and $p_k$
(\system addresses deadlock situation, i.e.,
where different platforms receive \zero messages in different order in the same way as SharPer \cite{amiri2019sharper}).
If the primary is not waiting for an uncommitted transaction,
it assigns sequence number $h_j$ to the message and multicasts a {\em signed} \one message
$\langle\langle\text{\scriptsize \ONE}, h_j, d \rangle_{\sigma_{\pi(p_j)}}, \mu \rangle$
to the nodes of its platform.
The primary node $\pi(p_j)$ also piggybacks the \zero message $\mu$ to its \one message to enable the
node to access the request and validate the \one message.

The primary node $\pi(p_j)$, as presented in line 13, multicasts a signed \two message
$\langle\text{\scriptsize \TWO}, h_i, h_j, d \rangle_{\sigma_{\pi(p_j)}}$
to {\em every} node of {\em all} involved platforms.

\begin{figure}[t] \center
\begin{minipage}{.4\textwidth}\centering
\includegraphics[width= \linewidth]{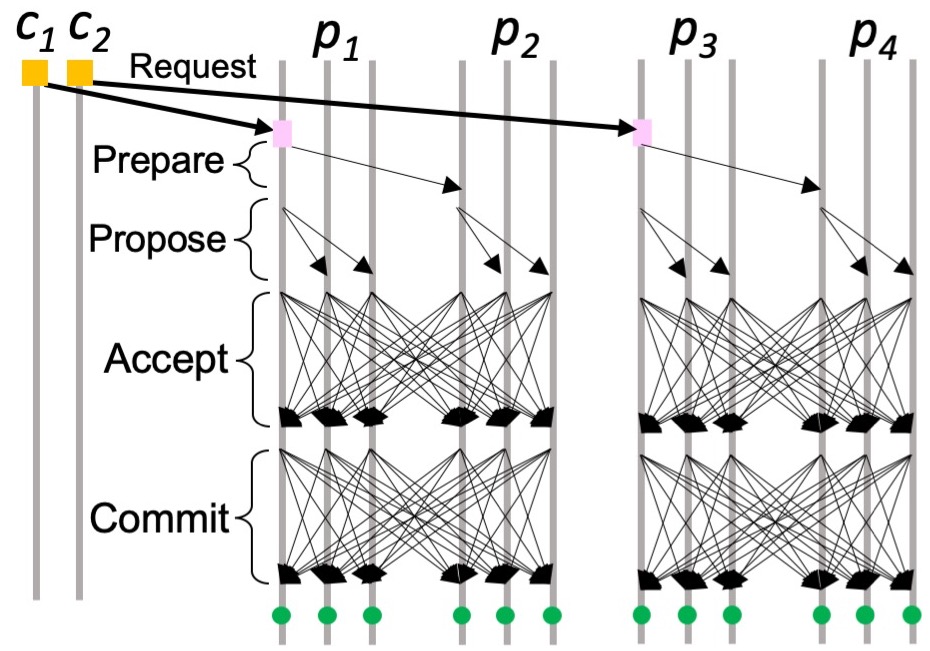}
\end{minipage}\hspace{2em}
\begin{minipage}{.46\textwidth}\centering
\includegraphics[width= \linewidth]{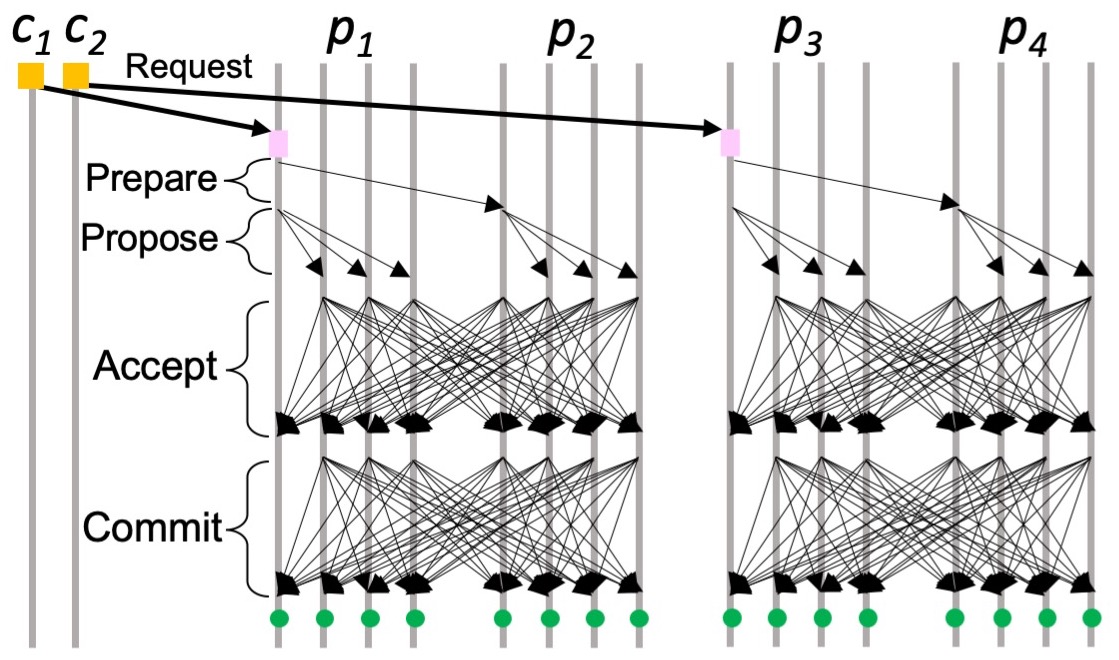}
\end{minipage}
\caption{Two concurrent cross-platform transaction flows for
(a) crash-only and (b) Byzantine nodes in \system where two disjoint sets of
platforms are involved in each task}
\label{fig:cross}
\end{figure}

Upon receiving a \one message
Once a node $r$ of an involved platform $p_j$ receives a \one message, as indicated in lines 8-10,
it validates the signature and message digest
(if the node belongs to the initiator platform  ($i=j$), it also checks $h_i$ to be valid (within a certain range))
since a malicious primary might multicast a request with an invalid sequence number.
In addition, if the node is currently involved in an uncommitted cross-platform request $m'$ where
the involved platforms of two requests $m$ and $m'$ overlap in some other platform, 
the node does not process the new request $m$ before the earlier request $m'$ is processed.
This is needed to ensure requests are committed in the same order on different platforms.
The node then multicasts a signed \two message including
the corresponding sequence number $h_j$ (that represents the order of $m$ in platform $p_j$),
and the digest $d=D(m)$ 
to {\em every} node of {\em all} involved platforms.
Note that if a platform behaves maliciously by not sending \two messages to other involved platforms,
the initiator platform can report the malicious behavior of the platform
by sending an ALERT message (similar to Section~\ref{sec:check})
to the registration authority resulting in imposing penalties.

As presented in lines 18-19,
each node waits for valid matching \two messages
from a local majority (i.e., either $f+1$ or $2f+1$ depending on the failure model) of {\em every} involved platform
with $h_i$ and $d$ that matches the \one message which was sent by primary $\pi(p_i)$.
We define the predicate {\sf accepted-local}$_{p_j}(m, h_i, h_j,r)$ to be true if and only if node $r$
has received the request $m$, a \one for $m$ with sequence number $h_i$ from the initiator platform $p_i$
and \two messages from a local majority of an involved platform $p_j$ that match the \one message.
The predicate {\sf accepted}$(m, h, r)$ where
$h = [h_i, h_j, ..., h_k]$
is then defined to be true on node $r$
if and only if {\sf accepted-local}$_{p_j}$ is true for {\em every} involved platform $p_j$ in cross-platform request $m$.
The order of sequence numbers in the predicate is an ascending order determined by their platform ids.
The \one and \two phases of the algorithm basically guarantee that non-faulty nodes agree on a total order
for the transactions.
When {\sf accepted}$(m, h, r)$ becomes true,
node $r$ multicasts a signed \three message
$\langle\text{\scriptsize \THREE}, h, d, r \rangle_{\sigma_r}$
to all nodes of every involved platforms.

Finally, as shown in lines 20-21, node $r$ waits for valid matching \three messages
from a local majority of {\em every} involved platform that matches its \three message.
The predicate {\sf committed-local}$_{p_j}(m, h, r)$
is defined to be true on node $r$ if and only if 
{\sf accepted}$(m, h, r)$ is true and node $r$ has accepted valid matching \three
messages from a local majority of platform $p_j$ that match the \one message for cross-platform transaction $m$.
The predicate {\sf committed}$(m, h, v, r)$ is then defined to be true on node $r$
if and only if {\sf committed-local}$_{p_j}$ is true for {\em every} involved platform $p_j$ in cross-platform transaction $m$.
The {\sf committed} predicate indeed shows that at least $f+1$ nodes of each involved platform have multicast valid \three messages.
When the {\sf committed} predicate becomes true, the node considers the transaction as committed.
If all transactions with lower sequence numbers than $h_j$ have already been committed,
the node appends a transaction block including the transaction
as well as the corresponding \three messages to its copy of the ledger.
Note that since {\sf \small commit} messages include the digest (cryptographic hash) of the corresponding transactions,
appending valid signed {\sf \small commit} messages to the blockchain ledger in addition to the transactions,
provides the same level of {\em immutability} guarantee as
including the cryptographic hash of the previous transaction in the transaction block, i.e.,
any attempt to alter the block data can easily be detected.
In terms of message complexity, \zero phase consists of $|P|$ messages, \one phase includes $n$ messages, and
\two and \three phases, each requires $n^2$ messages.

Figure~\ref{fig:cross} shows
the normal case operation for \system
to execute two concurrent cross-platform transactions in the presence of (a) crash-only and (b) Byzantine nodes where
each transaction accesses two disjoint platforms.
The network consists of four platforms where each platform includes either three or four nodes ($f=1$).

\begin{algorithm}[t]
\scriptsize
\caption{{\small Global Consensus}}
\label{alg:global}
\begin{algorithmic}[1]
\State {\bf init():} 
\State \quad $r$ := {\em node\_id}
\State \quad $p_i$ := the platform that initiates the consensus
\State \quad $\pi(p)$ := the primary node of platform $p$
\newline
\State {\bf upon receiving} valid transaction $m$ and $(r == \pi(p_i))$
\State \quad {\bf multicast} \azero to the primary node of every platform
\State \quad {\bf multicast} \aone to all nodes of $p_i$ 
\newline
\State {\bf upon receiving} valid $\mu{=}$ \azero and $r {==} \pi(p_j)$
\State \quad if $r$ is not involved in any uncommitted request $m'$
where $m$ and $m'$ intersect in some other platform $p_k$
\State \qquad {\bf multicast} \aonej to all nodes of $p_j$
\State \qquad {\bf multicast} \atwoj to all nodes
\newline
\State {\bf upon receiving} valid \aone and $r \in p_i$
\State \qquad {\bf multicast} \atwoi to all nodes
\newline
\State {\bf upon receiving} valid \aonej and $r \in p_j$
\State \qquad {\bf multicast} \atwo to all nodes
\newline
\State {\bf upon receiving} valid matching \atwo from {\em local-majority} of {\bf two-thirds} of platforms
\State \quad {\bf multicast} \athree to all nodes
\newline
\State {\bf upon receiving} valid \athree from {\em local-majority} of {\bf two-thirds} of platforms
\State \quad {\bf append} the transaction block to the ledger
\end{algorithmic}
\end{algorithm}

\subsubsection{Global Consensus}
The \ver transactions include group signatures and all tokens
that are consumed by different participants to perform a particular task.
In \system and in order to enable all platforms to check constraints,
\ver transactions are appended to the blockchains of all platforms.
To do so, a Byzantine fault-tolerant protocol is run among all nodes of every platform where
the protocol needs agreement from the {\em local majority} of the nodes of {\em two-thirds} of the platforms.
The local majority, similar to cross-platform consensus,
is defined based on the utilized consensus protocol within each platform.
However, there are two main differences between cross-platform consensus and global consensus.
First, in cross-platform consensus only the involved platforms participate, whereas,
in global consensus, every platform verifies transactions by checking the group signatures and consumed tokens.
Second, cross-platform consensus requires agreement from {\em every} platform, whereas,
in global consensus, agreement from only {\em two-thirds} of platforms is needed.
In fact, in cross-platform consensus, there might be some dependency between cross-platform transactions and internal ones,
thus, to ensure data consistency, every involved platform must agree on the order of the cross-platform transaction.
However, in global consensus, the goal is to verify the correctness of the transaction and as soon as
two-thirds of platforms verify that (assuming at most one-third of platforms might behave maliciously),
the transaction can be appended to the blockchain ledger.

\ifextend
Algorithm~\ref{alg:global} shows the normal case of {\em global consensus} in \system where
a Byzantine protocol is run among all nodes of every platform
(in contrast to cross-platform consensus where only the involved platforms participate).
The protocol, similar to cross-platform consensus,
process a transaction in four phases of \zero (lines 5-6), \one (lines 7-10), \two (lines 11-15), and \three (lines 16-19),
however, each node waits for matching \two and \three messages
from the local majority of only {\em two-thirds} of the platforms (as shown in lines 16 and 18).
\fi

Figure~\ref{fig:global} presents
the normal case operation of global consensus in \system. Here all platforms include crash-only nodes where $f=1$ and
the network consists of four platforms.

\begin{figure}[t] \center
\includegraphics[width= 0.4 \linewidth]{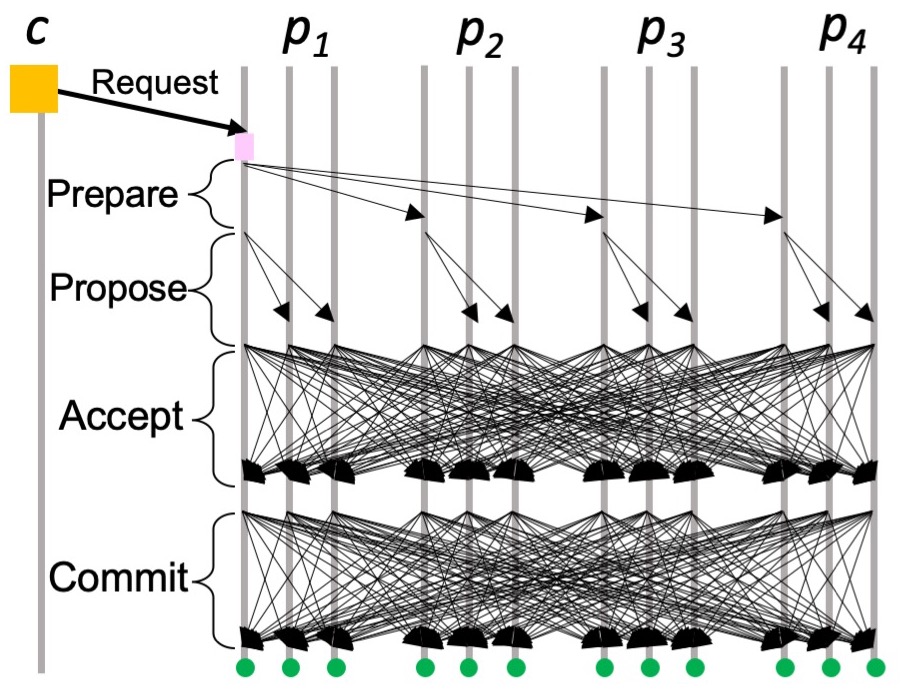}
\caption{Global consensus in \system}
\label{fig:global}
\end{figure}

\subsubsection{Primary Failure Handling and Deadlock situation}

In both cross-platform and global consensus protocols,
in addition to the normal case operation, \system has to deal with two other scenarios.
First, when the primary node fails.
Second, when nodes have not received \two messages from the
local-majority of either all involved platforms (cross-platform consensus)
or two-thirds of platforms (global consensus).
Indeed, the primary nodes of different platforms might multicast their \zero messages in parallel, hence,
different overlapping platforms might receive the messages in different order and
do not send \two messages for the second transaction in order to guarantee consistency.
We use similar techniques as SharPer \cite{amiri2019sharper} to address these two situations.
Note that for the local consensus, since the consensus protocol is pluggable,
\system follows the primary failure handling of the plugged protocol.
In addition, since only a single platform is involved in the local consensus, the second scenario
will not occur.

\medskip\noindent
{\bf Primary Failure Handling.}
The primary failure handling routine provides liveness by allowing the system to make progress when a primary fails.
The routine is triggered by timeout.
When node $r$ of some platform $p_j$ receives a valid \one message from its primary
for either an internal or a cross-platform transaction,
it starts a timer that expires after some predefined time $\tau$
(time $\tau$ is longer for cross-platform and \ver transactions).
If the timer has expired and the node has not committed the message,
the node suspects that the primary is faulty.
Similarly, if nodes of the initiator platform receive \two messages from other platforms with
sequence number that is different from their sequence number $h_i$
(that they have received from the initiator primary in the \one message),
they suspect that the initiator primary is faulty
(i.e., a malicious initiator primary might assign different sequence numbers to the same transaction
and send non-matching \zero messages to different platforms).
Note that nodes within other platforms are also able to detect conflicting messages
and send a \twoq message $\langle\text{\scriptsize \TWOQ}, h_i, h_j, d, r \rangle$ message
to every node of the initiator platform $p_i$ (the platform of the faulty primary).
When node $r$ (of the initiator platform) suspects that the initiator primary is faulty,
it initiates the routine by multicasting a \vchange message including all
received valid \one, \two, \twoq, and \three messages for all internal as well as cross-platform transactions
to every node within the platform.
An \twoq message is valid if it is received from at least $2f+1$ different nodes of another platform.
Upon receiving $2f$ \vchange message, the next primary (determined in the round robin manner based on the id of nodes)
handles the uncommitted transactions by multicasting a \newv message
including $2f+1$ \vchange messages and a \one message
for each uncommitted request to every node within the platform.
For uncommitted cross-platform and \ver transaction also, the new primary multicasts a \newv message
including $2f+1$ \vchange messages and related \zero messages
to the primary nodes of all involved platforms.

\medskip\noindent
{\bf Deadlock Situations.}
If the primary nodes of different platforms receive \zero messages
for different cross-platform transactions in conflicting orders,
the system might face a deadlock situation.
In particular, if two platforms $p_1$ and $p_2$ receive \zero messages for
cross-platform transactions $m$ and $m'$ in conflicting orders,
e.g., $p_1$ receives $m$ first and then $m'$ and $p_2$ receives $m'$ first and then $m$,
to ensure consistency property
(as explained in Algorithm~\ref{alg:cross}, line 11 and Algorithm~\ref{alg:global}, line 9),
they do not process the second transaction before committing the first one,
i.e., $p_1$ waits for the \three message of $m$ and $p_2$ waits the \three message of $m'$.
However, since committing a cross-platform transaction requires \two messages
from local majority of every involved platform, neither of $m$ and $m'$ can be committed (deadlock situation).
To resolve deadlock situation,
all platforms must reach to a unique order between conflicting messages
and based on that undo their sent \two messages if needed.
We explain the technique in two cases.
First, if both $m$ and $m'$ have been initiated by the same platform,
nodes of the other platforms, which are involved in both $m$ and $m'$, can detect the correct order
by comparing the sequence numbers of $m$ and $m'$ and in case
a node has already sent an \two message for the request with the higher sequence number,
it needs to undo its sent \two message by multicasting a \twof message (with the same structure as \two message)
with a different sequence number that is assigned by the primary node of the platform.
All nodes then remove the previous sent \two messages from their logs.
Second, when $m$ and $m'$ have been initiated by different platforms, e.g.,
$m$ is initiated by $p_3$ and $p_1$, $p_2$, and $p_3$ are involved in $m$ and
$m'$ is initiated by $p_4$ and $p_1$, $p_2$, and $p_4$ are involved in $m'$.
In such a situation and to determine a unique order, transactions $m$ and $m'$ are ordered 
based on the id of their initiator platforms.
As a result, if a node has already sent an \two message for the request with the higher initiator platform id,
it sends a \twof message to every node with a different sequence number (assigned by the primary of the platform).
All nodes also remove the previous sent \two messages from their logs.
Note that this technique can be used for the situations where
more than two platforms are in the intersection of concurrent requests or
when there are more than two requests with conflicting orders.

\subsubsection{Correctness Arguments}

\ifnextend
A consensus protocol has to satisfy agreement, validity,
consistency (total order), and termination properties \cite{cachin2011introduction}.
\fi
\ifextend
A consensus protocol has to satisfy four main properties \cite{cachin2011introduction}:
(1) {\em agreement:} every correct node must agree on the same value (Lemma~\ref{lmm:agree}),
(2) {\em Validity (integrity):} if a correct node commits a value,
then the value must have been proposed by some correct node (Lemma~\ref{lmm:val}),
(3) {\em Consistency (total order):} all correct nodes commit the same value in the same order (Lemma~\ref{lmm:cons}), and
(4) {\em termination:} eventually every node commits some value (Lemma~\ref{lmm:ter}).
\fi
The first three properties are known as {\em safety} and the termination property is known as {\em liveness}.
In an asynchronous system, where nodes can fail,
as shown by Fischer et al. \cite{fischer1985impossibility}, 
consensus has no solution that is both safe and live.
Therefore, \system guarantees safety in an asynchronous network, however,
similar to most fault-tolerant protocols, deals with termination (liveness) only during periods of synchrony using timers.
\ifnextend
All proofs are omitted due to space limitation and discussed in the full version of this paper \cite{amiri2020separ}.
\fi

\begin{lemma}\label{lmm:agree} (\textit{Agreement})
If node $r$ commits request $m$ with sequence number $h$,
no other correct node commits request $m'$ ($m \neq m'$) with the same sequence number $h$.
\end{lemma}

The \one and \two phases of both cross-platform and global consensus protocols
guarantee that correct nodes agree on a total order of requests.
Indeed, if the {\sf accepted}$(m, h, r)$ predicate where $h = [h_i, h_j, ..., h_k]$ is true, 
then {\sf accepted}$(m', h, q)$ is false for any non-faulty
node $q$ (including $r=q$) and any $m'$ such that $m \neq m'$.
This is true because $(m, h, r)$
implies that {\sf accepted-local}$_{p_j}(m, h_i, h_j,r)$ is true for each involved platform $p_j$ and
a local majority ($f+1$ crash-only or $2f+1$ Byzantine node) of platform $p_j$
have sent \two (or \one) messages
for request $m$ with sequence number $h_j$.
As a result, for {\sf accepted}$(m', h, q)$ to be true, at least one
non-faulty nodes needs to have sent two conflicting \two messages
with the same sequence number but different message digest.
This condition guarantees that first, a malicious primary cannot violate the safety and
second, at most one of the concurrent {\em conflicting} transactions, i.e.,
transactions that overlap in at least one platform,
can collect the required number of messages from each overlapping platform.

If the primary fails,
since (1) each committed request has been replicated on a quorum $Q_1$ of local majority nodes,
(2) to become primary agreement from a quorum of nodes is needed, and
(3) any two quorums (with size equal to local majority) intersect in at least one correct node 
that is aware of the latest committed request.
\system guarantees the agreement property for both internal as well as cross-platform transactions.

\begin{lemma}\label{lmm:val} (\textit{Validity})
If a correct node $r$ commits $m$, then $m$ must have been proposed by some correct node $\pi$.
\end{lemma}

\ifextend
In the presence of crash-only nodes, validity is ensured
since crash-only nodes do not send fictitious messages.
In the presence of Byzantine nodes, however, validity is guaranteed mainly based on
standard cryptographic assumptions about collision-resistant hashes, encryption, and signatures
which the adversary cannot subvert them.
Since the request as well as all messages are signed and
either the request or its digest is included in each message
(to prevent changes and alterations to any part of the message), and
in each step $2f+1$ matching messages (from each Byzantine platform) are required,
if a request is committed, the same request must have been proposed earlier.
\fi

\begin{lemma}\label{lmm:cons}(\textit{Consistency})
Let $P_\mu$ denote the set of involved platforms for a request $\mu$.
For any two committed requests $m$ and $m'$ and any two nodes $r_1$ and $r_2$
such that $r_1 \in p_i$, $r_2 \in p_j$, and $\{p_i,p_j\} \in P_m \cap P_{m'}$,
if $m$ is committed before $m'$ in $r_1$, then $m$ is committed before $m'$ in $r_2$.
\end{lemma}

As mentioned in both cross-platform and global consensuses,
once a node $r_1$ of some platform $p_i$ receives a \one message for some transaction $m$,
if the node is involved in some other uncommitted transaction $m'$
where $m$ and $m'$ overlap,
node $r_1$ does not send an \two message for transaction $m$ before $m'$ gets committed.
In this way, since committing request $m$ requires \two messages from a local majority of {\em every} (involved) platform,
$m$ cannot be committed until $m'$ is committed.
As a result the order of committing messages is the same in all involved platforms.
Note that ensuring consistency might result in deadlock situations which can be resolved
as explained earlier.

\begin{lemma}\label{lmm:ter}(\textit{Termination})
A request $m$ issued by a correct client eventually completes.
\end{lemma}

\ifextend
\system deals with termination (liveness) only during periods of synchrony using timers.
To do so, three scenarios need to be addressed.
If the primary is non-faulty and \zero messages are non-conflicting,
following the normal case operation of the protocol, request $m$ completes.
If the primary is non-faulty, but \zero (and as a result \two) messages are conflicting,
\system resolves the deadlock.
Finally, \system includes a routine to handle primary failures.
\fi
\section{Experimental Evaluations}\label{sec:eval}

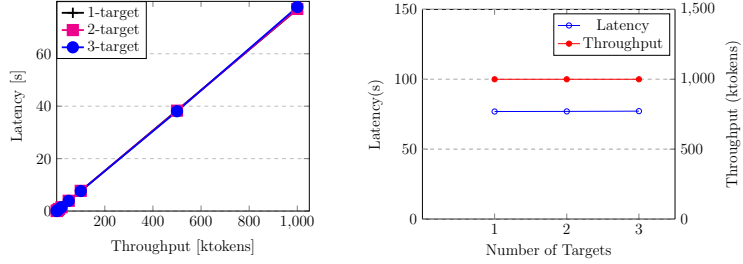
\begin{figure}[t]
\large
\centering
\begin{minipage}{.23\textwidth}\centering
\begin{tikzpicture}[scale=0.49]
\begin{axis}[
    xlabel={Throughput [ktokens]},
    ylabel={Latency [s]},
    xmin=0, xmax=1050,
    ymin=0, ymax=80,
    xtick={200,400,600,800,1000},
    ytick={0,20,40,60},
    legend columns=1,
    legend style={at={(axis cs:0,80)},anchor=north west}, 
    ymajorgrids=true,
    grid style=dashed,
]

 \addplot[
    color=black,
    mark=+,
    mark size=4pt,
    line width=0.5mm,
    ]
    coordinates {
    (0.1,0.015)(1,0.091)(5,0.40)(10,0.78)(20,1.551)(50,3.871)(100,7.626)(500,38.421)(1000,76.952)};

\addplot[
    color=magenta,
    mark= square*,
    mark size=4pt,
    line width=0.5mm,
    ]
    coordinates {
    (0.1,0.015)(1,0.089)(5,0.401)(10,0.782)(20,1.463)(50,3.913)(100,7.705)(500,38.292)(1000,77.103)};

\addplot[
    color=blue,
    mark=*,
    mark size=4pt,
    line width=0.5mm,
    ]
    coordinates {
    (0.1,0.016)(1,0.089)(5,0.408)(10,0.789)(20,1.463)(50,3.901)(100,7.641)(500,38.082)(1000,77.845)};

\addlegendentry{$1$-target}
\addlegendentry{$2$-target}
\addlegendentry{$3$-target}
 
\end{axis}
\end{tikzpicture}
\end{minipage}\hspace{2em}
\begin{minipage}{.23\textwidth}\centering
\begin{tikzpicture}[scale=0.49]
\pgfplotsset{
    xmin=0, xmax=3.5,
    xtick={1,2,3},
        ymajorgrids=true,
    grid style=dashed,
}

\begin{axis}[
  axis y line*=left,
  ymin=0, ymax=150,
  xlabel=Number of Targets,
  ylabel=Latency(s)
]

\addplot[smooth,mark=o,blue]
  coordinates{(1,76.931)(2,77.019)(3,77.212)};
\label{plot_two}

\end{axis}

\begin{axis}[
  axis y line*=right,
  axis x line=none,
  ymin=0, ymax=1500,
  ylabel=Throughput (ktokens)
]
\addlegendimage{/pgfplots/refstyle=plot_two}
\addlegendentry{Latency}
\addplot[smooth,mark=*,red]
coordinates{(1,1000)(2,1000)(3,1000)};
\addlegendentry{Throughput}
\end{axis}
\end{tikzpicture}
\end{minipage}
\caption{Varying the Number of Tokens}
  \label{fig:token}
\end{figure}

In this section, we conduct several experiments to evaluate both the scalability of \system and the overhead of privacy.
We consider a complex heavily loaded setting with $4$ platforms (except for the last experiment),
$20000$ requesters and $20000$ workers where
both requesters and workers are randomly registered to one or more platforms.
Once a task is submitted to a platform, the platform randomly assigns the task to
one or more (depending on the required number of contributions) idle workers (to avoid any delay).
The experiments consist of two main parts.
In the first part (Sections~\ref{sec:token} and \ref{sec:expconst}),
the privacy costs of \system (i.e., tokens and regulations) is evaluated, whereas
in the second part (Sections~\ref{sec:expcross} and \ref{sec:expplatform}),
the scalability of \system is evaluated.
For the purpose of this evaluation, and as explained earlier,
we do not focus on the description of tasks and contributions
(both are modeled as arbitrary bitstrings).
In addition, $v$-tokens, as explained earlier, are very similar to 
$e$-tokens except for the private part that has no significant impact on the performance and the
number of interaction phases which is even less than  $e$-tokens. %
Therefore, we only focus on $e$-tokens (i.e., enforceable regulations) in the experiments.
To implement group signatures,
we use the protocol proposed in~\cite{camenisch2004group}.
The experiments were conducted on the Amazon EC2 platform.
Each VM is a c4.2xlarge instance with 8 vCPUs and 15GB RAM,
Intel Xeon E5-2666 v3 processor clocked at 3.50 GHz.
When reporting throughput measurements, we use an increasing
number of tasks submitted by requesters running on a single VM,
until the end-to-end throughput is saturated,
and state the throughput and latency just below saturation.

\subsection{Token Generation}
\label{sec:token}

In the first set of experiments, we measure the performance of token generation (performed by RA)
in \system for different classes of regulations.
We consider different classes of regulations, i.e.,
single-target (e.g., $((w, *, *), <, \theta)$),
two-target (e.g., $((*, p, r), <, \theta)$), and
three-target (e.g., $((w, p, r), <, \theta)$) enforceable regulations.
As shown in Figure~\ref{fig:token}(a),
Our experiments show that
\system is able to generate tokens in linear time.
\system generates each token in $~0.07 ms$, hence, generating
$1$ million tokens in $\sim 76$ seconds.
This clearly demonstrates the scalability of token generation especially since
token generation is executed periodically, e.g., every week or every month.
Note that, we use a single machine to generate tokens, however,
tokens related to different regulations can be generated in parallel.
Hence, the throughput of \system can linearly increase by
running the token generation routine on multiple machines, e.g.,
with $10$ machines, \system is able to generate $1$ million tokens in $\sim 8$ seconds.
Moreover, to provide fault tolerance, the tokens can be replicated on multiple machines
following the replication techniques.
Furthermore,
as can be seen in Figure~\ref{fig:token}(b),
classes of regulations (i.e., the number of targets)
does not affect the performance, i.e., the throughput and latency of token generation
are constant in terms of the number of targets.
It should, however, be noted that a regulations with more targets
requires more tokens to be generated, e.g.,
$((w, *, *), <, \theta)$ requires $|\mathcal{W}| * \theta$ tokens, whereas
$((w, p, r), <, \theta)$ requires $|\mathcal{W}| * |\mathcal{P}| * |\mathcal{R}| * \theta$ tokens to be generated.

\subsection{Classes of Regulations}
\label{sec:expconst}

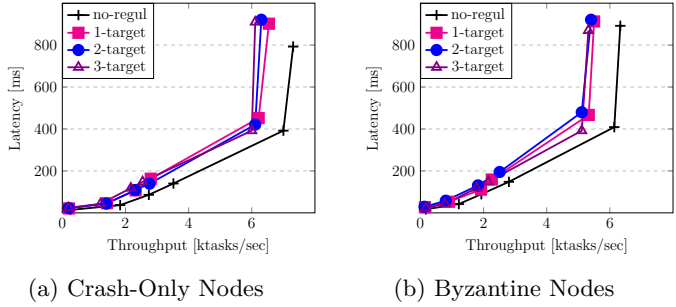
\begin{figure}[t]
\large
\centering
\begin{minipage}{.23\textwidth}\centering
\begin{tikzpicture}[scale=0.49]
\begin{axis}[
    xlabel={Throughput [ktasks/sec]},
    ylabel={Latency [ms]},
    xmin=0, xmax=8,
    ymin=0, ymax=1000,
    xtick={0,2,4,6},
    ytick={200,400,600,800},
    legend columns=1,
    legend style={at={(axis cs:0,1000)},anchor=north west}, 
    ymajorgrids=true,
    grid style=dashed,
]

 \addplot[
    color=black,
    mark=+,
    mark size=4pt,
    line width=0.5mm,
    ]
    coordinates {
    (0.258,14)(1.832,37)(2.74,86)(3.519,140)(7.001,392)(7.311,793)};

\addplot[
    color=magenta,
    mark= square*,
    mark size=4pt,
    line width=0.5mm,
    ]
    coordinates {
    (0.208,21)(1.413,46)(2.32,107)(2.801,163)(6.219,452)(6.542,902)};

\addplot[
    color=blue,
    mark=*,
    mark size=4pt,
    line width=0.5mm,
    ]
    coordinates {
    (0.202,23)(1.392,44)(2.31,109)(2.763,139)(6.109,421)(6.302,921)};

\addplot[
    color=violet,
    mark=triangle,
    mark size=4pt,
    line width=0.5mm,
    ]
    coordinates {
    (0.192,24)(1.208,43)(2.17,118)(2.542,148)(6.001,392)(6.110,910)};

\addlegendentry{no-regul}
\addlegendentry{1-target}
\addlegendentry{2-target}
\addlegendentry{3-target}
 
\end{axis}
\end{tikzpicture}
{\footnotesize (a) Crash-Only Nodes}
\end{minipage}\hspace{2em}
\begin{minipage}{.23\textwidth} \centering
\begin{tikzpicture}[scale=0.49]
\begin{axis}[
    xlabel={Throughput [ktasks/sec]},
    ylabel={Latency [ms]},
    xmin=0, xmax=8,
    ymin=0, ymax=1000,
    xtick={0,2,4,6},
    ytick={200,400,600,800},
    legend columns=1,
    legend style={at={(axis cs:0,1000)},anchor=north west}, 
    ymajorgrids=true,
    grid style=dashed,
]

 \addplot[
    color=black,
    mark=+,
    mark size=4pt,
    line width=0.5mm,
    ]
    coordinates {
    (0.172,17)(1.216,42)(1.93,89)(2.802,148)(6.14,409)(6.33,892)};

\addplot[
    color=magenta,
    mark= square*,
    mark size=4pt,
    line width=0.5mm,
    ]
    coordinates {
    (0.142,27)(0.902,53)(1.92,111)(2.261,159)(5.331,467)(5.501,913)};

\addplot[
    color=blue,
    mark=*,
    mark size=4pt,
    line width=0.5mm,
    ]
    coordinates {
    (0.132,29)(0.803,58)(1.823,131)(2.510,195)(5.109,480)(5.403,921)};

\addplot[
    color=violet,
    mark=triangle,
    mark size=4pt,
    line width=0.5mm,
    ]
    coordinates {
    (0.173,24)(0.780,43)(1.78,118)(2.151,148)(5.110,392)(5.319,870)};

\addlegendentry{no-regul}
\addlegendentry{1-target}
\addlegendentry{2-target}
\addlegendentry{3-target}

\end{axis}
\end{tikzpicture}
{\footnotesize (b) Byzantine Nodes}
\end{minipage}
\caption{Varying the Class of regulations}
\label{fig:const}
\end{figure}

In the second set of experiments, we measure
the overhead of privacy-preserving techniques, e.g., group signatures and tokens, used in \system.
We consider the basic scenario with no regulation
(i.e., no need to exchange and validate tokens and signatures)
and compare it with
three different scenarios where each task has to satisfy
a single target,
a two-target, and
a three-target regulation.
The system consists of four platforms and the workload includes
$90\%$ internal and $10\%$ cross-platform tasks
(the typical settings in partitioned databases \cite{thomson2012calvin, taft2014store}) where
two (randomly chosen) platforms are involved in \sub and \clm transactions of cross-platform tasks
(note that all platforms are still involved in the \ver transaction of each task).
We also assume that completion of each task requires a single contribution
(and obviously a \sub and a \ver transaction).

When nodes follow the crash failure model and the system has no regulations,
as can be seen in Figure~\ref{fig:const}(a), \system is able to process $7000$ tasks with $390$ ms latency
before the end-to-end throughput is saturated (the penultimate point).
Adding regulations, results in more phases of communication between different participants
to exchange tokens and signatures, however,
\system is still able to process $6200$ tasks (the penultimate point) with $450$ ms latency.
This result demonstrates that all privacy-preserving techniques that are used in \system
results in only $11\%$ and $15\%$ overhead in terms of the throughput and latency respectively.
Moreover, the class of regulations does not significantly affect the performance of \system.
This is expected because more targets result in only
increasing the number of (parallel) tokens and signature exchanges
while the number of communication phases is not affected.

Similarly, in the presence of Byzantine nodes and as shown in Figure~\ref{fig:const}(b),
\system is able to process $6140$ tasks with $409$ ms latency with no regulations and
$5331$ tasks ($13\%$ overhead)  with $467$ ms ($14\%$ overhead) latency with single-target regulations.
As before, the class of regulations does not affect the performance.

It should be noted that by increasing the number of regulations, \system
still demonstrates similar behavior as shown in this experiment.
Indeed, adding more regulations,
while it results in adding more tokens and possibly more participants and signatures,
it does not affect the consensus protocols and other communication phases.

\subsection{Cross-Platform Tasks}
\label{sec:expcross}

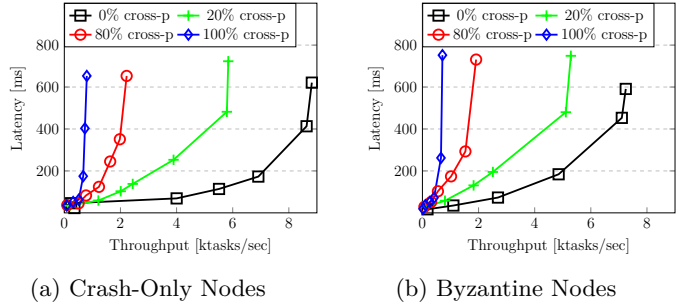
\begin{figure}[t]
\large
\centering
\begin{minipage}{.23\textwidth}\centering
\begin{tikzpicture}[scale=0.49]
\begin{axis}[
    xlabel={Throughput [ktasks/sec]},
    ylabel={Latency [ms]},
    xmin=0, xmax=9,
    ymin=0, ymax=1000,
    xtick={0,2,4,6,8},
    ytick={200,400,600,800},
    legend columns=2,
    legend style={at={(axis cs:0,1000)},anchor=north west}, 
    ymajorgrids=true,
    grid style=dashed,
]

 \addplot[
    color=black,
    mark=square,
    mark size=4pt,
    line width=0.5mm,
    ]
    coordinates {
    (0.359,23)(0.2432,46)(4.002,69)(5.512,114)(6.913,173)(8.632,413)(8.819,621)};

\addplot[
    color=green,
    mark=+,
    mark size=4pt,
    line width=0.5mm,
    ]
    coordinates {
    (0.195,34)(1.213,59)(2.01,103)(2.437,138)(3.892,253)(5.792,482)(5.842,723)};

 \addplot[
     color=red,
     mark=o,
     mark size=4pt,
     line width=0.5mm,
     ]
     coordinates {
    (0.105,37)(0.509,45)(0.773,83)(1.234,124)(1.634,245)(1.978,351)(2.221,653)};

\addplot[
    color=blue,
    mark=diamond,
    mark size=4pt,
    line width=0.5mm,
    ]
    coordinates {
    (0.087,35)(0.315,52)(0.514,63)(0.673,175)(0.734,403)(0.801,653)};

\addlegendentry{$0\%$ cross-p}
\addlegendentry{$20\%$ cross-p}
\addlegendentry{$80\%$ cross-p}
\addlegendentry{$100\%$ cross-p}
 
\end{axis}
\end{tikzpicture}
{\footnotesize (a) Crash-Only Nodes}
\end{minipage}\hspace{2em}
\begin{minipage}{.23\textwidth} \centering
\begin{tikzpicture}[scale=0.49]
\begin{axis}[
    xlabel={Throughput [ktasks/sec]},
    ylabel={Latency [ms]},
    xmin=0, xmax=9,
    ymin=0, ymax=1000,
    xtick={0,2,4,6,8},
    ytick={200,400,600,800},
    legend columns=2,
    legend style={at={(axis cs:0,1000)},anchor=north west},
    ymajorgrids=true,
    grid style=dashed,
]

 \addplot[
    color=black,
    mark=square,
    mark size=4pt,
    line width=0.5mm,
    ]
    coordinates {
    (0.180,17)(1.103,35)(2.692,73)(4.853,184)(7.104,453)(7.245,591)};

\addplot[
    color=green,
    mark=+,
    mark size=4pt,
    line width=0.5mm,
    ]
    coordinates {
    (0.132,29)(0.803,58)(1.823,131)(2.510,195)(5.109,480)(5.289,749)};

\addplot[
    color=red,
    mark=o,
    mark size=4pt,
    line width=0.5mm,
    ]
    coordinates {
    (0.071,29)(0.313,52)(0.552,103)(1.019,174)(1.531,293)(1.913,731)};

\addplot[
    color=blue,
    mark=diamond,
    mark size=4pt,
    line width=0.5mm,
    ]
    coordinates {
    (0.021,19)(0.105,37)(0.27,53)(0.413,72)(0.661,262)(0.713,752)};

\addlegendentry{$0\%$ cross-p}
\addlegendentry{$20\%$ cross-p}
\addlegendentry{$80\%$ cross-p}
\addlegendentry{$100\%$ cross-p}

\end{axis}
\end{tikzpicture}
{\footnotesize (b) Byzantine Nodes}
\end{minipage}
\caption{Varying Number of Cross-Platform Tasks}
\label{fig:crossexp}
\end{figure}

In the next set of experiments, we measure the performance of \system for
workloads with different percentages of cross-platform tasks.
We consider four different workloads with
(1) no cross-platform tasks,
(2) $20\%$ cross-platform tasks,
(3) $80\%$ cross-platform tasks, and
(4) $100\%$ cross-platform tasks.
As before, completion of each task requires a single contribution and
two (randomly chosen) platforms are involved in each cross-platform task.
The system includes four platforms and each task has to satisfy two randomly chosen regulations.
We consider two different networks with crash-only and Byzantine nodes.
When all nodes follow crash-only nodes, as presented in Figure~\ref{fig:crossexp}(a),
\system is able to process $8600$ tasks with $400$ ms latency 
before the end-to-end throughput is saturated (the penultimate point)
if all tasks are internal.
Note that even when all tasks are internal,
the \ver transaction of each task still needs global consensus among all platforms.
Increasing the percentage of cross-platform tasks to $20\%$,
reduces the overall throughput to $5800$ ($67\%$) with $400$ ms latency since
processing cross-platform tasks requires cross-platform consensus.
By increasing the percentage of cross-platform tasks to $80\%$ and then $100\%$,
the throughput of \system will reduce to $1900$ and $700$ with the same ($400$ ms) latency.
This is expected because when most tasks are cross-platform ones,
more nodes are involved in processing a task and more messages are exchanged.
In addition, the possibility of parallel processing of tasks will be significantly reduced.
In the presence of Byzantine nodes, as shown in Figure~\ref{fig:crossexp}(b),
\system demonstrates a similar behavior as the crash-only case.
When all tasks are internal, \system processes $7100$ tasks with $450$ ms latency.
Increasing the percentage of cross-platform tasks to $20\%$ and $80\%$ will
reduce the throughput to $4900$ and $1700$ tasks with the same ($450$ ms) latency respectively.
Finally, when all tasks are cross-platform, \system is able to process $700$ tasks with $450$ ms latency.

\subsection{The Number of Platforms}
\label{sec:expplatform}

In the last set of experiments, we measure the scalability of \system in crowdworking environments
with different number of platforms ($1$ to $5$ platforms) for
both crash-only and Byzantine nodes (assuming $f=1$ in each platform).
The networks, thus, include
$3$, $6$, $9$, $12$, and $15$ crash-only nodes or $4$, $8$, $12$, $16$ and $20$ Byzantine nodes.
Each task has to satisfy two randomly chosen regulations,
two (randomly chosen) platforms are involved in each cross-platform tasks,
completion of each task requires a single contribution,
and the workloads include $90\%$ internal and $10\%$ cross-platform tasks.
Note that in the scenario with a single platform, all tasks are internal.
As shown in Figure~\ref{fig:platform}(a), in the presence of crash-only nodes,
the performance of the system is improved by adding more platforms, e.g.,
with five platform, \system processes $6600$ tasks with $400$ ms latency whereas
in a single platform setting, \system processes $3300$ task with the same latency.
While adding more platforms improves the performance of \system, the relation between
the increased number of platforms and the improved throughput is non-linear
(the number of platforms has been increased $5$ times while the throughput doubled).
This is expected because
adding more platforms while increases the possibility of parallel processing of internal tasks,
makes the global consensus algorithm (which is needed for every single task) more expensive.
In the presence of Byzantine nodes, \system demonstrates similar behavior, e.g.,
processes $5500$ tasks with $470$ ms latency with $5$ platforms.

\begin{figure}[t]
\large
\centering
\begin{minipage}{.23\textwidth}\centering
\begin{tikzpicture}[scale=0.49]
\begin{axis}[
    xlabel={Throughput [ktasks/sec]},
    ylabel={Latency [ms]},
    xmin=0, xmax=7,
    ymin=0, ymax=700,
    xtick={0,1,2,3,4,5,6},
    ytick={150,300,450,600},
    legend columns=1,
    legend style={at={(axis cs:0,700)},anchor=north west}, 
    ymajorgrids=true,
    grid style=dashed,
]

 \addplot[
    color=black,
    mark=+,
    mark size=4pt,
    line width=0.5mm,
    ]
    coordinates {
    (0.111,20)(0.719,41)(1.209,63)(1.612,98)(2.921,282)(3.331,392)(3.416,471)};

\addplot[
    color=magenta,
    mark= square*,
    mark size=4pt,
    line width=0.5mm,
    ]
    coordinates {
    (0.151,22)(0.901,42)(1.613,77)(2.211,119)(3.205,221)(4.632,370)(4.819,442)};

\addplot[
    color=blue,
    mark=*,
    mark size=4pt,
    line width=0.5mm,
    ]
    coordinates {

    (0.183,23)(1.161,44)(2.001,84)(2.381,103)(3.902,223)(5.421,396)(5.611,511)};

\addplot[
    color=violet,
    mark=triangle,
    mark size=4pt,
    line width=0.5mm,
    ]
    coordinates {
    (0.217,25)(1.412,47)(2.309,77)(2.771,101)(4.621,239)(6.063,409)(6.101,536)};
    
\addplot[
    color=green,
    mark=o,
    mark size=4pt,
    line width=0.5mm,
    ]
    coordinates {
    (0.249,29)(1.721,53)(2.693,76)(3.701,112)(4.910,218)(6.621,432)(6.713,552)};

\addlegendentry{$3$ nodes}
\addlegendentry{$6$ nodes}
\addlegendentry{$9$ nodes}
\addlegendentry{$12$ nodes}
\addlegendentry{$15$ nodes}

\end{axis}
\end{tikzpicture}
{\footnotesize (a) Crash-Only Nodes}
\end{minipage}\hspace{2em}
\begin{minipage}{.23\textwidth} \centering
\begin{tikzpicture}[scale=0.49]
\begin{axis}[
    xlabel={Throughput [ktasks/sec]},
    ylabel={Latency [ms]},
    xmin=0, xmax=7,
    ymin=0, ymax=700,
    xtick={0,1,2,3,4,5,6},
    ytick={150,300,450,600},
   legend columns=1,
   legend style={at={(axis cs:0,700)},anchor=north west}, 
    ymajorgrids=true,
    grid style=dashed,
]

 \addplot[
    color=black,
    mark=+,
    mark size=4pt,
    line width=0.5mm,
    ]
    coordinates {
    (0.061,15)(0.392,32)(0.92,67)(1.319,103)(2.401,312)(2.623,479)};

\addplot[
    color=magenta,
    mark=square*,
    mark size=4pt,
    line width=0.5mm,
    ]
    coordinates {
    (0.085,21)(0.614,41)(1.208,78)(1.813,144)(2.221,203)(2.911,326)(3.302,554)};

\addplot[
    color=blue,
    mark=*,
    mark size=4pt,
    line width=0.5mm,
    ]
    coordinates {
    (0.109,36)(0.703,61)(1.529,101)(2.17,156)(2.842,231)(3.982,402)(4.490,623)};

\addplot[
    color=violet,
    mark=triangle,
    mark size=4pt,
    line width=0.5mm,
    ]
    coordinates {
    (0.128,31)(0.811,55)(1.782,102)(2.591,162)(3.742,278)(5.001,463)(5.11,631)};

\addplot[
    color=green,
    mark=o,
    mark size=4pt,
    line width=0.5mm,
    ]
    coordinates {
    (0.119,29)(1.042,44)(2.202,87)(3.001,140)(4.105,260)(5.519,471)(5.593,620)};

\addlegendentry{$4$ nodes}
\addlegendentry{$8$ nodes}
\addlegendentry{$12$ nodes}
\addlegendentry{$16$ nodes}
\addlegendentry{$20$ nodes}

\end{axis}
\end{tikzpicture}
{\footnotesize (b) Byzantine Nodes}
\end{minipage}
\caption{Varying the Number of Platforms}
  \label{fig:platform}
\end{figure}
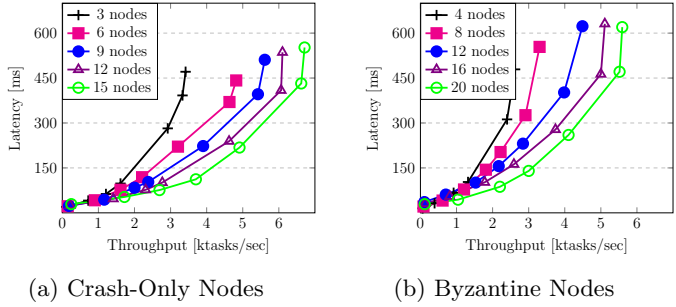

\section{Related Work}\label{sec:related}

Enhancing privacy in the context of crowdworking has been addressed by several recent studies
with various kinds of guarantees, from differential privacy ~\cite{to2018privacy,to2016differentially} to
cryptography \cite{liu2017protecting,liu2017privacy,liu2018efficient},
mostly focusing on spatial crowdsourcing and the use of geolocation to perform assignments. 
In ZebraLancer~\cite{lu2018zebralancer} and ZKCrowd~\cite{zhu2019zkcrowd}, blockchain
is also used to add transparency guarantees on top of privacy. 
However, all these works consider a single-platform context. 
Other works propose to manage multiple platforms in crowdsourcing. For instance, Fluid~\cite{han2019fluid} proposes to share workers' profiles between multiple platforms, and Wang et al. \cite{wang2019interaction} provides an interesting insight on federated recommendation systems for workers and the balance of surplus of tasks or workers in a cross-platform context. 
However, none of them provides tools to manage external regulations: to the best of our knowledge, \system is the first to support a multi-platform crowdworking context,
with external regulations, transparency, and privacy expectations at the same time. 

In the context of permissioned blockchains, 
Hyperledger Fabric \cite{androulaki2018hyperledger} ensures data confidentiality
using Private Data Collections \cite{collections}.
Private Data Collections manage confidential data that two or more entities
want to keep private from others.
Quorum \cite{morgan2016quorum}
supports public and private transactions and ensures
the confidentiality of private transactions using the Zero-knowledge proof technique.
Quorum, however, orders all public and as private transactions using
a single consensus protocol resulting in low throughput.

\system is inspired by various permissioned blockchain systems,
and more specifically by SharPer \cite{amiri2019sharper},
a permissioned blockchain system that uses sharding to improve scalability.
Unlike \system, SharPer is designed for environments with a single application enterprise, 
which would correspond to the ledger of a single platform in \system.
\ifextend
Furthermore, SharPer supports simple transactions that are all processed
in the same way, while \system has different types (i.e., \sub, \clm, and \ver)
and depending on its type, the transactions will be processed in different ways,
e.g., \ver transactions require coordination among all platforms.
\fi
In SharPer, the infrastructure consists of clusters of nodes where
the single ledger is partitioned into shards that are assigned to different clusters.
\ifextend
Within each cluster, each data shard is replicated on the nodes of the cluster.
A permissioned blockchain is used to incorporate all transactions making any changes to the data shards.
SharPer supports both intra-shard and cross-shard transactions and
incorporates a flattened protocol to establish consensus 
among clusters on the order of cross-shard transactions.
\fi
\system, on the other hand, includes multiple enterprises (i.e., platforms)
that might not trust each other.
As a result, while each shard in SharPer is similar to a platform in \system,
any cross-platform communication in \system requires a Byzantine fault-tolerant protocol
that tolerates malicious behavior at the platform level. 
This is in contrast to SharPer where the type of cross-platform consensus
depends on the failure model of nodes (and not the application, which is always trusted).
\system needs to support, in addition to local and cross-platform consensus, a global consensus protocol
that establishes agreement among all (possibly untrusted) platforms on the order of \ver transactions.
Finally, while in SharPer all nodes are either crash-only or Byzantine,
in \system, platforms are run by different enterprises and hence can have different failure models.

Providing anonymity as well as untraceability has been addressed by ZCash~\cite{zcash}
which is restricted to the management of crypto-currency issues.
Hawk~\cite{hawk} and Raziel~\cite{raziel} manage wider issues, and include general smart contracts.
However, these solutions do not incorporate infrastructures with multiple platforms,
nor implement regulations (let alone anonymized ones). 
Finally, Solidus~\cite{solidus} proposes to privately manage a multi-platform banking system,
with individual banks managing their own clients while allowing cross-platform transactions.
While Solidus may be sufficient for banking systems,
it does not consider users that subscribe to multiple platforms, nor envisions global profiles or regulations.
\section{Conclusion}\label{sec:conc}

In this paper, we present an overall vision
for future of work multi-platform crowdworking environments consisting of three main dimensions:
regulations, security and architecture.
We then introduce \system, 
(to the best of our knowledge)
the first to address the problem of enforcing global regulations
over multi-platform crowdworking environments.
\system enables official institutions to express global regulations in simple and unambiguous terms,
guarantees the satisfaction of global regulations by construction, and
allows participants to prove to external entities their involvement in crowdworking tasks,
all in a privacy-preserving manner.
\system also uses transparent blockchain ledgers shared across multiple platforms and
enables collaboration among platforms through a suite of distributed consensus protocols.
We prove the privacy requirements of \system and
conduct an extensive experimental evaluation to measure the performance and scalability of \system.
\balance

\bibliographystyle{abbrv}
\bibliography{main}

\end{document}